\documentclass[11pt,a4paper]{article}
\usepackage[margin=23mm]{geometry}
\usepackage{lmodern}
\usepackage{amsmath,amssymb,amsthm,mathtools,bm}
\usepackage{booktabs,longtable,array}
\usepackage{colortbl}
\usepackage{graphicx}
\usepackage{xcolor}
\usepackage{float}
\usepackage{url}
\usepackage{etoolbox}
\usepackage{hyperref}
\usepackage{bookmark}
\hypersetup{
  unicode=true,
  pdfencoding=auto,
  psdextra=true,
  colorlinks=true,
  linkcolor=blue!55!black,
  urlcolor=blue!55!black,
  citecolor=blue!55!black
}
\newcolumntype{L}[1]{>{\raggedright\arraybackslash}p{#1}}

\newcommand{\F}{\mathbb{F}}
\newcommand{\HX}{H_X}
\newcommand{\HZ}{H_Z}
\newcommand{\HXL}{H_X^{\mathrm{lift}}}
\newcommand{\HZL}{H_Z^{\mathrm{lift}}}
\newcommand{\vect}[1]{\bm{#1}}
\newcommand{\rank}{\operatorname{rank}}
\newcommand{\wt}{\operatorname{wt}}
\newcommand{\row}{\operatorname{row}}
\newcommand{\kerop}{\operatorname{ker}}

\newtheoremstyle{uprighttheorem}
  {0.45em}{0.45em}{\normalfont}{}{\bfseries}{.}{0.5em}{}
\theoremstyle{uprighttheorem}
\newtheorem{definition}{Definition}
\newtheorem{theorem}{Theorem}

\newtheorem{proposition}{Proposition}
\newtheorem{example}{Example}
\AtEndEnvironment{definition}{\hfill\ensuremath{\square}}
\AtEndEnvironment{theorem}{\hfill\ensuremath{\square}}
\AtEndEnvironment{lemma}{\hfill\ensuremath{\square}}
\AtEndEnvironment{proposition}{\hfill\ensuremath{\square}}
\AtEndEnvironment{example}{\hfill\ensuremath{\square}}
\title{A Two-Branch Finite-Field Construction for Regular CSS LDPC Bases}
\author{Koki Okada and Kenta Kasai\\Institute of Science Tokyo\\\texttt{okada.k.3154@m.isct.ac.jp}\\\texttt{kenta@ict.eng.isct.ac.jp}}
\date{}

\begin{document}
\maketitle

\begin{abstract}
This paper develops a two-branch multiplicative-coset construction for regular Calderbank--Shor--Steane (CSS) quantum low-density parity-check base matrices.
For a target column weight \(J\) and an even row weight \(L\), the method reduces regularity, CSS orthogonality, and same-type 4-cycle exclusion to explicit quotient-coset conditions over a finite field.
A normalized exhaustive search for these conditions produces base matrices for several \((J,L)\) pairs, so the construction is not tied to a single degree distribution.
The construction separates the finite-length design into two stages: the base matrix fixes the degree distribution and the first girth constraints, and a cyclic lift randomizes edge connections subject to exact algebraic checks.
As a detailed example, we carry one \((3,10)\)-regular base through the lift and decoding stages.
For this example, the selected 64-fold lift gives a code whose same-type Tanner graphs have girth at least eight, and it also excludes a specified weight-16 nondegenerate logical-support orbit.
Complete support enumeration excludes all nontrivial logical representatives of weight at most 16 on both sides. The resulting instance is a \([[10240,4108,\,18\le d\le32]]\) CSS code.
For decoding, we use joint log-domain belief propagation together with low-complexity deterministic post-processing rules for small residual syndromes, including repairs for residual patterns with two unsatisfied checks.
The frame error rate (FER) measurements provide finite-length decoding data for this detailed example; at depolarizing probability \(p=0.058\), the post-processing FER is \(1.0\times10^{-7}\).
\end{abstract}

\section{Introduction}
Low-density parity-check (LDPC) codes were introduced as sparse parity-check codes for classical error correction \cite{gallager1962ldpc}.
Their Tanner graphs support iterative message-passing decoders, and this sparse-graph viewpoint became one of the standard foundations of modern coding theory.
Classical finite-length LDPC design can be viewed as the simultaneous control of three quantities: degree distribution, girth, and randomness in the edge connections.
For a fixed channel and update rule, the variable-node and check-node degree distributions determine the belief-propagation (BP) threshold predicted by density evolution \cite{richardson2001capacity}.
This made degree-distribution design a central part of LDPC code design, especially for irregular ensembles approaching capacity \cite{richardson2001design}.
At the same time, finite-length decoding is sensitive to short cycles, so the girth and the number of short cycles are also structural design targets.
Finally, a favorable degree distribution and a large local girth are not enough if the finite graph still contains deterministic short-cycle or low-weight structures that degrade decoding.
The classical ensemble argument therefore relies on randomness in the edge connections: random edge connections, or equivalently random lifts of a base graph, reduce deterministic finite-length structures and make the finite graph locally resemble the ensemble used in BP threshold analysis \cite{kelley2008ldpc}.
The same requirement appears in recent quantum LDPC construction work: after the degree distribution and orthogonality constraints are fixed, additional randomization or randomized local modification is used to avoid deterministic finite-length structures while preserving the sparse CSS constraints \cite{okada2025random}.
Quantum error correction introduced stabilizer codes and, within them, Calderbank--Shor--Steane (CSS) codes \cite{calderbank1996good,steane1996multiple}.
In the CSS setting, one does not choose a single sparse parity-check matrix.
Instead, one chooses two sparse binary check matrices whose row spaces must satisfy a commutation condition.
Recent finite-length quantum LDPC work has studied this constraint from several complementary directions, including non-binary extensions, permutation-based high-girth constructions, degeneracy-aware decoding, and CSS constructions with performance close to coding-theoretic or density-evolution reference points \cite{komoto2025codingbound,kasai2025errorfloor,kasai2025degeneracy,kasai2026orthogonality}.

Quantum LDPC codes inherit the sparse-graph motivation of classical LDPC coding, but they also inherit this additional commutation requirement.
Early sparse-graph constructions and iterative decoders established the basic finite-length setting \cite{mackay2004sparse,poulin2008iterative}.
At the same time, explicit algebraic constructions already showed that commutation is difficult to reconcile with the graph properties that are desirable for iterative decoding, such as regularity and the absence of short cycles \cite{hagiwara2007qc,aly2007finitegeometry}.
This issue has continued in more recent quasi-cyclic, affine-permutation, and non-binary CSS constructions, where the support structure, field labels, and permutation choices are designed so that orthogonality and short-cycle constraints remain compatible \cite{komoto2025columnweight2,kasai2025girth16apm,kasai2025nonbinaryextension}.
This commutation constraint remains a central design restriction.
Thus, even before optimizing a finite-length decoder implementation, one has to decide which degree distribution to realize and then construct a sparse CSS matrix pair that realizes it without violating commutation.
In the present design, randomness is introduced in two controlled stages.
At the base-matrix stage, the finite-field construction gives many coefficient choices, and the search selects a base that has the target degrees, CSS orthogonality, and no same-type 4-cycles.
At the lift stage, cyclic lift labels are sampled and then checked so that the remaining base 6-cycles do not close and the specified low-weight support pattern is excluded.
This separates the roles of the two stages: the base fixes the degree distribution and the first girth constraints, while the lift supplies randomized edge connections subject to exact algebraic checks.

One major line of work addressed asymptotic parameters.
Hypergraph-product codes gave a systematic positive-rate construction with square-root distance scaling \cite{tillich2014quantum}.
Subsequent product and lifted constructions improved the asymptotic picture, first by exceeding square-root distance scaling and then by producing asymptotically good quantum LDPC codes \cite{hastings2021quantum_ldpc,panteleev2022asymptotically,leverrier2022quantum}.
Recent fixed-degree random CSS constructions also connect the sparse quantum LDPC setting with Gilbert--Varshamov-type distance guarantees \cite{kasai2026gv}.
These developments established that sparse quantum codes can simultaneously achieve strong asymptotic rate and distance guarantees.

Another line of work remained focused on finite block lengths, where the main questions are different.
At finite length, one must specify an actual sparse matrix pair, verify the commutation condition, control short cycles in the Tanner graphs, identify low-weight logical operators or certified candidates, and evaluate a concrete decoder.
In this regime, decoder behavior depends on the interaction between graph structure, degeneracy, and post-processing rules \cite{panteleev2021degenerate,breuckmann2021qldpc}.
For that reason, finite-length construction work often proceeds by combining algebraic design, explicit certification, and numerical decoding experiments.

The present paper belongs to this finite-length line.
Recent work around this finite-length program separates several tasks that are often combined in direct sparse CSS design: constructing a regular commuting base, adding randomness without destroying orthogonality, certifying short-cycle and low-weight structures, and evaluating belief-propagation-based decoding on a concrete lifted code \cite{okada2025random,kasai2026orthogonality,kasai2026heuristic,okada2026highgirth,okada2026squarebase,kasai2026factorgraph}.
The decoding side of the same program includes empirical studies of sharp finite-length transitions and residual error-floor events under joint BP decoding \cite{komoto2025sharp}.
These cited finite-length constructions provide the immediate context for the present paper, but the specific design used here is different from the previously studied non-binary, affine-permutation, affine-coset, and square-base hypergraph-product constructions.

The detailed example targets a CSS code with column weight three and row weight ten.
This degree pair has different constraints from the previously studied cases with the same column weight and row weights eight and twelve.
Compared with row weight eight, row weight ten requires more incidences in each check while still controlling same-type short cycles.
Compared with row weight twelve, row weight ten leaves less freedom for arranging the paired cross-type overlaps that enforce CSS orthogonality.
The search for this case led to a more general finite-field construction of regular CSS bases.
Thus the main design problem in this paper is not only to realize one base with column weight three and row weight ten, but to give a verifiable base-construction principle for many choices of column and row weights.

Our solution is a two-branch multiplicative-coset construction over a finite field.
The construction is stated in a field-independent form first, so that the reason for regularity and the sufficient coset conditions used to certify orthogonality and same-type 4-cycle exclusion are explicit.
The resulting quotient-coset conditions can be checked by a finite normalized search, and the paper records coefficient examples for several regularities.
Among these bases, we select the instance with column weight three and row weight ten as the detailed example and carry it through the finite-length stages: specialization to the field with sixteen elements, selection of a 64-fold cyclic lift, exclusion of the remaining same-type base 6-cycles, and exclusion of a specified weight-16 logical-support orbit.
The decoding section then studies this concrete lifted code under joint BP and deterministic post-processing, using the factor-graph formulation in \cite{kasai2026factorgraph}.

The position of this paper in the historical development is therefore specific.
It does not propose a new asymptotic family, and it does not claim an exact distance theorem for a broad class of codes.
Instead, it gives a general finite-field base construction, shows that the construction produces many regular CSS bases, and then treats the base with column weight three and row weight ten as the detailed finite-length example.
Choosing lifts, certifying distance, and measuring finite-length decoding performance for the other regularities are left as future work.

\section{General Theory of the Two-Branch Finite-Field Base}\label{sec:general-theory}
This section gives the field-independent part of the construction.
The purpose is to define the two-branch finite-field incidence rule and to record the general certificates used later.
The main outputs of the section are \((J,L)\) regularity, a sufficient coset condition for CSS orthogonality, and a separate sufficient coset condition for excluding same-type 4-cycles.

\subsection{Base Incidence Rule}\label{subsec:base-incidence-rule}

The first construction stage builds the base matrices $\HX$ and $\HZ$.
This is the part of the design where the incidence pattern is chosen explicitly.
Before specializing to a concrete field, we record the general incidence rule that specifies the rows, columns, and nonzero positions of the two base matrices.
The construction uses translations of a multiplicative subgroup in a finite field to place the ones in $\HX$ and $\HZ$.
Two branches are used because CSS orthogonality is enforced by pairing the $X$-$Z$ overlaps of branch \(0\) and branch \(1\), while same-type 4-cycles are avoided by requiring the corresponding same-type difference cosets to be disjoint.
The column weight is denoted by $J$, and the row weight is denoted by $L$.
The row weight is controlled by the subgroup size: each branch contributes one incident column for each element of $M$, so the two-branch construction gives
\[
  L=2|M|.
\]
Thus, to target an even row weight $L$, one seeks a finite field whose multiplicative group contains a subgroup of order $L/2$.

\begin{definition}[Two-branch multiplicative-coset base]\label{def:two-branch-coset-base}
Let $\F$ be a finite field and let $M$ be a multiplicative subgroup of $\F^\times$.
Fix a branch index $\lambda\in\{0,1\}$ and a column weight $J$.
For each branch, choose coefficients
\[
  a_0^{(\lambda)},\ldots,a_{J-1}^{(\lambda)}\in\F,\qquad
  b_0^{(\lambda)},\ldots,b_{J-1}^{(\lambda)}\in\F.
\]
The columns are indexed by
\[
  (\lambda,t,h)\in\{0,1\}\times\F\times M.
\]
We call $\lambda$ the branch index, $t$ the translation coordinate, and $h$ the subgroup coordinate of the base column.
The $X$ rows and $Z$ rows are indexed by
\[
  (i,r),\quad i=0,\ldots,J-1,\ r\in\F,
\]
and
\[
  (j,s),\quad j=0,\ldots,J-1,\ s\in\F,
\]
respectively.
Here $i$ and $j$ are row-group indices, while $r$ and $s$ are field-position coordinates.
Define binary matrices $\HX$ and $\HZ$ with these row and column index sets as follows.
For brevity, put
\[
  m_0=m_X=m_Z=J|\F|.
\]
Both matrices have $m_0$ rows and
\[
  2|\F||M|
\]
columns.
When the matrices are displayed or stored as ordinary arrays, fix enumerations
\[
  \F=\{\alpha_0,\ldots,\alpha_{q-1}\},\qquad
  M=\{\mu_0,\ldots,\mu_{m-1}\},
\]
where \(q=|\F|\) and \(m=|M|\), and order the branch indices as
\(0,1\).
The zero-based global row and column coordinates are then
\begin{equation}
\label{eq:global-row-coordinate}
  \rho_X(i,\alpha_u)=iq+u,\qquad
  \rho_Z(j,\alpha_u)=jq+u,
\end{equation}
and
\begin{equation}
\label{eq:global-column-coordinate}
  \kappa(\lambda,\alpha_u,\mu_v)
  =\lambda qm+um+v.
\end{equation}
For one-based printed row or column numbers, add \(1\) to these coordinates.
In $\HX$, column $(\lambda,t,h)$ has ones in the rows
\begin{equation}
\label{eq:hx-incidence}
  (i,\ t+a_i^{(\lambda)}h)\qquad (i=0,\ldots,J-1)
\end{equation}
and zeros elsewhere.
In $\HZ$, column $(\lambda,t,h)$ has ones in the rows
\begin{equation}
\label{eq:hz-incidence}
  (j,\ t+b_j^{(\lambda)}h)\qquad (j=0,\ldots,J-1)
\end{equation}
and zeros elsewhere.
Equivalently, these equations specify the $X$-side and $Z$-side incidences of the base Tanner graphs.
\end{definition}

The next example is included only to make the indexing convention concrete.
It is not used as a later code instance; it shows how the abstract row, column, and subgroup coordinates become ordinary binary matrices.

\begin{example}[Small \(J=3\) two-branch base]\label{ex:small-two-branch-base}
The following small certified instance has \(J=3\), matching the column weight used in the detailed \((3,10)\) example.
Take
\[
  \F=\F_7,\qquad M=\{1,2,4\}\subset\F_7^\times,\qquad J=3.
\]
Set
\[
  \bm a^{(0)}=(0,1,3),\quad
  \bm b^{(0)}=(2,4,5),\quad
  \bm a^{(1)}=(0,3,1),\quad
  \bm b^{(1)}=(4,2,5),
\]
with entries interpreted modulo \(7\).
Then \(m_0=21\), the base length is \(42\), and the row weight is \(L=2|M|=6\).
These coefficients satisfy the coset certificates in Theorems~\ref{thm:orthogonality-coset-condition} and~\ref{thm:same-type-4cycle-coset-condition}.
For example, consider the column indexed by \((\lambda,t,h)=(0,0,1)\).
Using the incidence rule in Definition~\ref{def:two-branch-coset-base}, this column has ones in $\HX$ at
\[
  (0,0+0\cdot1),\quad (1,0+1\cdot1),\quad (2,0+3\cdot1),
\]
namely at the rows \((0,0),(1,1),(2,3)\).
In the displayed \(\HX^{\mathrm{ex}}\) matrix below, this is the first column, and the corresponding entries are highlighted.
The same column has ones in $\HZ$ at
\[
  (0,0+2\cdot1),\quad (1,0+4\cdot1),\quad (2,0+5\cdot1),
\]
namely at the rows \((0,2),(1,4),(2,5)\).
In the displayed \(\HZ^{\mathrm{ex}}\) matrix, the corresponding entries in the same first column are highlighted.
For the displayed arrays, use the enumerations
\(\alpha_u=u\) for \(\F_7\) and
\((\mu_0,\mu_1,\mu_2)=(1,2,4)\) for \(M\).
With the global coordinates in \eqref{eq:global-row-coordinate} and
\eqref{eq:global-column-coordinate}, the resulting binary matrices are
\begin{center}
\begingroup
\scriptsize
\setlength{\arraycolsep}{0.58pt}
\renewcommand{\arraystretch}{0.70}
\newcommand{\mzero}{\textcolor{black!15}{0}}
\newcommand{\mxone}{\cellcolor[HTML]{1F77B4}\textcolor{white}{1}}
\newcommand{\myone}{\cellcolor[HTML]{E07A2F}\textcolor{white}{1}}
\newcommand{\mzone}{\cellcolor[HTML]{2CA25F}\textcolor{white}{1}}
\newcommand{\mhit}{\cellcolor[HTML]{FFE066}\textcolor{black}{\mathbf{1}}}
\newcommand{\msep}{\vrule width 0.25pt}
\newcommand{\mlayersep}{\vrule width 0.8pt}
\[
\HX^{\mathrm{ex}}=
\left[
\begin{array}{ccc!{\msep}ccc!{\msep}ccc!{\msep}ccc!{\msep}ccc!{\msep}ccc!{\msep}ccc!{\mlayersep}ccc!{\msep}ccc!{\msep}ccc!{\msep}ccc!{\msep}ccc!{\msep}ccc!{\msep}ccc}
\mhit&\mxone&\mxone&\mzero&\mzero&\mzero&\mzero&\mzero&\mzero&\mzero&\mzero&\mzero&\mzero&\mzero&\mzero&\mzero&\mzero&\mzero&\mzero&\mzero&\mzero&\mxone&\mxone&\mxone&\mzero&\mzero&\mzero&\mzero&\mzero&\mzero&\mzero&\mzero&\mzero&\mzero&\mzero&\mzero&\mzero&\mzero&\mzero&\mzero&\mzero&\mzero\\
\mzero&\mzero&\mzero&\mxone&\mxone&\mxone&\mzero&\mzero&\mzero&\mzero&\mzero&\mzero&\mzero&\mzero&\mzero&\mzero&\mzero&\mzero&\mzero&\mzero&\mzero&\mzero&\mzero&\mzero&\mxone&\mxone&\mxone&\mzero&\mzero&\mzero&\mzero&\mzero&\mzero&\mzero&\mzero&\mzero&\mzero&\mzero&\mzero&\mzero&\mzero&\mzero\\
\mzero&\mzero&\mzero&\mzero&\mzero&\mzero&\mxone&\mxone&\mxone&\mzero&\mzero&\mzero&\mzero&\mzero&\mzero&\mzero&\mzero&\mzero&\mzero&\mzero&\mzero&\mzero&\mzero&\mzero&\mzero&\mzero&\mzero&\mxone&\mxone&\mxone&\mzero&\mzero&\mzero&\mzero&\mzero&\mzero&\mzero&\mzero&\mzero&\mzero&\mzero&\mzero\\
\mzero&\mzero&\mzero&\mzero&\mzero&\mzero&\mzero&\mzero&\mzero&\mxone&\mxone&\mxone&\mzero&\mzero&\mzero&\mzero&\mzero&\mzero&\mzero&\mzero&\mzero&\mzero&\mzero&\mzero&\mzero&\mzero&\mzero&\mzero&\mzero&\mzero&\mxone&\mxone&\mxone&\mzero&\mzero&\mzero&\mzero&\mzero&\mzero&\mzero&\mzero&\mzero\\
\mzero&\mzero&\mzero&\mzero&\mzero&\mzero&\mzero&\mzero&\mzero&\mzero&\mzero&\mzero&\mxone&\mxone&\mxone&\mzero&\mzero&\mzero&\mzero&\mzero&\mzero&\mzero&\mzero&\mzero&\mzero&\mzero&\mzero&\mzero&\mzero&\mzero&\mzero&\mzero&\mzero&\mxone&\mxone&\mxone&\mzero&\mzero&\mzero&\mzero&\mzero&\mzero\\
\mzero&\mzero&\mzero&\mzero&\mzero&\mzero&\mzero&\mzero&\mzero&\mzero&\mzero&\mzero&\mzero&\mzero&\mzero&\mxone&\mxone&\mxone&\mzero&\mzero&\mzero&\mzero&\mzero&\mzero&\mzero&\mzero&\mzero&\mzero&\mzero&\mzero&\mzero&\mzero&\mzero&\mzero&\mzero&\mzero&\mxone&\mxone&\mxone&\mzero&\mzero&\mzero\\
\mzero&\mzero&\mzero&\mzero&\mzero&\mzero&\mzero&\mzero&\mzero&\mzero&\mzero&\mzero&\mzero&\mzero&\mzero&\mzero&\mzero&\mzero&\mxone&\mxone&\mxone&\mzero&\mzero&\mzero&\mzero&\mzero&\mzero&\mzero&\mzero&\mzero&\mzero&\mzero&\mzero&\mzero&\mzero&\mzero&\mzero&\mzero&\mzero&\mxone&\mxone&\mxone\\ \noalign{\hrule height 0.8pt}
\mzero&\mzero&\mzero&\mzero&\mzero&\mzero&\mzero&\mzero&\mzero&\mzero&\mzero&\myone&\mzero&\mzero&\mzero&\mzero&\myone&\mzero&\myone&\mzero&\mzero&\mzero&\mzero&\mzero&\mzero&\myone&\mzero&\mzero&\mzero&\myone&\mzero&\mzero&\mzero&\myone&\mzero&\mzero&\mzero&\mzero&\mzero&\mzero&\mzero&\mzero\\
\mhit&\mzero&\mzero&\mzero&\mzero&\mzero&\mzero&\mzero&\mzero&\mzero&\mzero&\mzero&\mzero&\mzero&\myone&\mzero&\mzero&\mzero&\mzero&\myone&\mzero&\mzero&\mzero&\mzero&\mzero&\mzero&\mzero&\mzero&\myone&\mzero&\mzero&\mzero&\myone&\mzero&\mzero&\mzero&\myone&\mzero&\mzero&\mzero&\mzero&\mzero\\
\mzero&\myone&\mzero&\myone&\mzero&\mzero&\mzero&\mzero&\mzero&\mzero&\mzero&\mzero&\mzero&\mzero&\mzero&\mzero&\mzero&\myone&\mzero&\mzero&\mzero&\mzero&\mzero&\mzero&\mzero&\mzero&\mzero&\mzero&\mzero&\mzero&\mzero&\myone&\mzero&\mzero&\mzero&\myone&\mzero&\mzero&\mzero&\myone&\mzero&\mzero\\
\mzero&\mzero&\mzero&\mzero&\myone&\mzero&\myone&\mzero&\mzero&\mzero&\mzero&\mzero&\mzero&\mzero&\mzero&\mzero&\mzero&\mzero&\mzero&\mzero&\myone&\myone&\mzero&\mzero&\mzero&\mzero&\mzero&\mzero&\mzero&\mzero&\mzero&\mzero&\mzero&\mzero&\myone&\mzero&\mzero&\mzero&\myone&\mzero&\mzero&\mzero\\
\mzero&\mzero&\myone&\mzero&\mzero&\mzero&\mzero&\myone&\mzero&\myone&\mzero&\mzero&\mzero&\mzero&\mzero&\mzero&\mzero&\mzero&\mzero&\mzero&\mzero&\mzero&\mzero&\mzero&\myone&\mzero&\mzero&\mzero&\mzero&\mzero&\mzero&\mzero&\mzero&\mzero&\mzero&\mzero&\mzero&\myone&\mzero&\mzero&\mzero&\myone\\
\mzero&\mzero&\mzero&\mzero&\mzero&\myone&\mzero&\mzero&\mzero&\mzero&\myone&\mzero&\myone&\mzero&\mzero&\mzero&\mzero&\mzero&\mzero&\mzero&\mzero&\mzero&\mzero&\myone&\mzero&\mzero&\mzero&\myone&\mzero&\mzero&\mzero&\mzero&\mzero&\mzero&\mzero&\mzero&\mzero&\mzero&\mzero&\mzero&\myone&\mzero\\
\mzero&\mzero&\mzero&\mzero&\mzero&\mzero&\mzero&\mzero&\myone&\mzero&\mzero&\mzero&\mzero&\myone&\mzero&\myone&\mzero&\mzero&\mzero&\mzero&\mzero&\mzero&\myone&\mzero&\mzero&\mzero&\myone&\mzero&\mzero&\mzero&\myone&\mzero&\mzero&\mzero&\mzero&\mzero&\mzero&\mzero&\mzero&\mzero&\mzero&\mzero\\ \noalign{\hrule height 0.8pt}
\mzero&\mzero&\mzero&\mzero&\mzone&\mzero&\mzero&\mzero&\mzone&\mzero&\mzero&\mzero&\mzone&\mzero&\mzero&\mzero&\mzero&\mzero&\mzero&\mzero&\mzero&\mzero&\mzero&\mzero&\mzero&\mzero&\mzero&\mzero&\mzero&\mzero&\mzero&\mzero&\mzone&\mzero&\mzero&\mzero&\mzero&\mzone&\mzero&\mzone&\mzero&\mzero\\
\mzero&\mzero&\mzero&\mzero&\mzero&\mzero&\mzero&\mzone&\mzero&\mzero&\mzero&\mzone&\mzero&\mzero&\mzero&\mzone&\mzero&\mzero&\mzero&\mzero&\mzero&\mzone&\mzero&\mzero&\mzero&\mzero&\mzero&\mzero&\mzero&\mzero&\mzero&\mzero&\mzero&\mzero&\mzero&\mzone&\mzero&\mzero&\mzero&\mzero&\mzone&\mzero\\
\mzero&\mzero&\mzero&\mzero&\mzero&\mzero&\mzero&\mzero&\mzero&\mzero&\mzone&\mzero&\mzero&\mzero&\mzone&\mzero&\mzero&\mzero&\mzone&\mzero&\mzero&\mzero&\mzone&\mzero&\mzone&\mzero&\mzero&\mzero&\mzero&\mzero&\mzero&\mzero&\mzero&\mzero&\mzero&\mzero&\mzero&\mzero&\mzone&\mzero&\mzero&\mzero\\
\mhit&\mzero&\mzero&\mzero&\mzero&\mzero&\mzero&\mzero&\mzero&\mzero&\mzero&\mzero&\mzero&\mzone&\mzero&\mzero&\mzero&\mzone&\mzero&\mzero&\mzero&\mzero&\mzero&\mzero&\mzero&\mzone&\mzero&\mzone&\mzero&\mzero&\mzero&\mzero&\mzero&\mzero&\mzero&\mzero&\mzero&\mzero&\mzero&\mzero&\mzero&\mzone\\
\mzero&\mzero&\mzero&\mzone&\mzero&\mzero&\mzero&\mzero&\mzero&\mzero&\mzero&\mzero&\mzero&\mzero&\mzero&\mzero&\mzone&\mzero&\mzero&\mzero&\mzone&\mzero&\mzero&\mzone&\mzero&\mzero&\mzero&\mzero&\mzone&\mzero&\mzone&\mzero&\mzero&\mzero&\mzero&\mzero&\mzero&\mzero&\mzero&\mzero&\mzero&\mzero\\
\mzero&\mzero&\mzone&\mzero&\mzero&\mzero&\mzone&\mzero&\mzero&\mzero&\mzero&\mzero&\mzero&\mzero&\mzero&\mzero&\mzero&\mzero&\mzero&\mzone&\mzero&\mzero&\mzero&\mzero&\mzero&\mzero&\mzone&\mzero&\mzero&\mzero&\mzero&\mzone&\mzero&\mzone&\mzero&\mzero&\mzero&\mzero&\mzero&\mzero&\mzero&\mzero\\
\mzero&\mzone&\mzero&\mzero&\mzero&\mzone&\mzero&\mzero&\mzero&\mzone&\mzero&\mzero&\mzero&\mzero&\mzero&\mzero&\mzero&\mzero&\mzero&\mzero&\mzero&\mzero&\mzero&\mzero&\mzero&\mzero&\mzero&\mzero&\mzero&\mzone&\mzero&\mzero&\mzero&\mzero&\mzone&\mzero&\mzone&\mzero&\mzero&\mzero&\mzero&\mzero
\end{array}
\right],
\]

\vspace{-0.2em}

\[
\HZ^{\mathrm{ex}}=
\left[
\begin{array}{ccc!{\msep}ccc!{\msep}ccc!{\msep}ccc!{\msep}ccc!{\msep}ccc!{\msep}ccc!{\mlayersep}ccc!{\msep}ccc!{\msep}ccc!{\msep}ccc!{\msep}ccc!{\msep}ccc!{\msep}ccc}
\mzero&\mzero&\mzero&\mzero&\mzero&\mzero&\mzero&\mzero&\mzero&\mzero&\mxone&\mzero&\mzero&\mzero&\mzero&\mxone&\mzero&\mzero&\mzero&\mzero&\mxone&\mzero&\mzero&\mzero&\mzero&\mzero&\mzero&\mzero&\mzero&\mzero&\mxone&\mzero&\mzero&\mzero&\mzero&\mzero&\mzero&\mzero&\mxone&\mzero&\mxone&\mzero\\
\mzero&\mzero&\mxone&\mzero&\mzero&\mzero&\mzero&\mzero&\mzero&\mzero&\mzero&\mzero&\mzero&\mxone&\mzero&\mzero&\mzero&\mzero&\mxone&\mzero&\mzero&\mzero&\mxone&\mzero&\mzero&\mzero&\mzero&\mzero&\mzero&\mzero&\mzero&\mzero&\mzero&\mxone&\mzero&\mzero&\mzero&\mzero&\mzero&\mzero&\mzero&\mxone\\
\mhit&\mzero&\mzero&\mzero&\mzero&\mxone&\mzero&\mzero&\mzero&\mzero&\mzero&\mzero&\mzero&\mzero&\mzero&\mzero&\mxone&\mzero&\mzero&\mzero&\mzero&\mzero&\mzero&\mxone&\mzero&\mxone&\mzero&\mzero&\mzero&\mzero&\mzero&\mzero&\mzero&\mzero&\mzero&\mzero&\mxone&\mzero&\mzero&\mzero&\mzero&\mzero\\
\mzero&\mzero&\mzero&\mxone&\mzero&\mzero&\mzero&\mzero&\mxone&\mzero&\mzero&\mzero&\mzero&\mzero&\mzero&\mzero&\mzero&\mzero&\mzero&\mxone&\mzero&\mzero&\mzero&\mzero&\mzero&\mzero&\mxone&\mzero&\mxone&\mzero&\mzero&\mzero&\mzero&\mzero&\mzero&\mzero&\mzero&\mzero&\mzero&\mxone&\mzero&\mzero\\
\mzero&\mxone&\mzero&\mzero&\mzero&\mzero&\mxone&\mzero&\mzero&\mzero&\mzero&\mxone&\mzero&\mzero&\mzero&\mzero&\mzero&\mzero&\mzero&\mzero&\mzero&\mxone&\mzero&\mzero&\mzero&\mzero&\mzero&\mzero&\mzero&\mxone&\mzero&\mxone&\mzero&\mzero&\mzero&\mzero&\mzero&\mzero&\mzero&\mzero&\mzero&\mzero\\
\mzero&\mzero&\mzero&\mzero&\mxone&\mzero&\mzero&\mzero&\mzero&\mxone&\mzero&\mzero&\mzero&\mzero&\mxone&\mzero&\mzero&\mzero&\mzero&\mzero&\mzero&\mzero&\mzero&\mzero&\mxone&\mzero&\mzero&\mzero&\mzero&\mzero&\mzero&\mzero&\mxone&\mzero&\mxone&\mzero&\mzero&\mzero&\mzero&\mzero&\mzero&\mzero\\
\mzero&\mzero&\mzero&\mzero&\mzero&\mzero&\mzero&\mxone&\mzero&\mzero&\mzero&\mzero&\mxone&\mzero&\mzero&\mzero&\mzero&\mxone&\mzero&\mzero&\mzero&\mzero&\mzero&\mzero&\mzero&\mzero&\mzero&\mxone&\mzero&\mzero&\mzero&\mzero&\mzero&\mzero&\mzero&\mxone&\mzero&\mxone&\mzero&\mzero&\mzero&\mzero\\ \noalign{\hrule height 0.8pt}
\mzero&\mzero&\mzero&\mzero&\mzero&\mzero&\mzero&\mzero&\mzero&\myone&\mzero&\mzero&\mzero&\mzero&\mzero&\mzero&\mzero&\myone&\mzero&\myone&\mzero&\mzero&\mzero&\mzero&\mzero&\mzero&\mzero&\mzero&\mzero&\mzero&\mzero&\myone&\mzero&\mzero&\mzero&\mzero&\myone&\mzero&\mzero&\mzero&\mzero&\myone\\
\mzero&\myone&\mzero&\mzero&\mzero&\mzero&\mzero&\mzero&\mzero&\mzero&\mzero&\mzero&\myone&\mzero&\mzero&\mzero&\mzero&\mzero&\mzero&\mzero&\myone&\mzero&\mzero&\myone&\mzero&\mzero&\mzero&\mzero&\mzero&\mzero&\mzero&\mzero&\mzero&\mzero&\myone&\mzero&\mzero&\mzero&\mzero&\myone&\mzero&\mzero\\
\mzero&\mzero&\myone&\mzero&\myone&\mzero&\mzero&\mzero&\mzero&\mzero&\mzero&\mzero&\mzero&\mzero&\mzero&\myone&\mzero&\mzero&\mzero&\mzero&\mzero&\myone&\mzero&\mzero&\mzero&\mzero&\myone&\mzero&\mzero&\mzero&\mzero&\mzero&\mzero&\mzero&\mzero&\mzero&\mzero&\myone&\mzero&\mzero&\mzero&\mzero\\
\mzero&\mzero&\mzero&\mzero&\mzero&\myone&\mzero&\myone&\mzero&\mzero&\mzero&\mzero&\mzero&\mzero&\mzero&\mzero&\mzero&\mzero&\myone&\mzero&\mzero&\mzero&\mzero&\mzero&\myone&\mzero&\mzero&\mzero&\mzero&\myone&\mzero&\mzero&\mzero&\mzero&\mzero&\mzero&\mzero&\mzero&\mzero&\mzero&\myone&\mzero\\
\mhit&\mzero&\mzero&\mzero&\mzero&\mzero&\mzero&\mzero&\myone&\mzero&\myone&\mzero&\mzero&\mzero&\mzero&\mzero&\mzero&\mzero&\mzero&\mzero&\mzero&\mzero&\myone&\mzero&\mzero&\mzero&\mzero&\myone&\mzero&\mzero&\mzero&\mzero&\myone&\mzero&\mzero&\mzero&\mzero&\mzero&\mzero&\mzero&\mzero&\mzero\\
\mzero&\mzero&\mzero&\myone&\mzero&\mzero&\mzero&\mzero&\mzero&\mzero&\mzero&\myone&\mzero&\myone&\mzero&\mzero&\mzero&\mzero&\mzero&\mzero&\mzero&\mzero&\mzero&\mzero&\mzero&\myone&\mzero&\mzero&\mzero&\mzero&\myone&\mzero&\mzero&\mzero&\mzero&\myone&\mzero&\mzero&\mzero&\mzero&\mzero&\mzero\\
\mzero&\mzero&\mzero&\mzero&\mzero&\mzero&\myone&\mzero&\mzero&\mzero&\mzero&\mzero&\mzero&\mzero&\myone&\mzero&\myone&\mzero&\mzero&\mzero&\mzero&\mzero&\mzero&\mzero&\mzero&\mzero&\mzero&\mzero&\myone&\mzero&\mzero&\mzero&\mzero&\myone&\mzero&\mzero&\mzero&\mzero&\myone&\mzero&\mzero&\mzero\\ \noalign{\hrule height 0.8pt}
\mzero&\mzero&\mzero&\mzero&\mzero&\mzone&\mzone&\mzero&\mzero&\mzero&\mzero&\mzero&\mzero&\mzone&\mzero&\mzero&\mzero&\mzero&\mzero&\mzero&\mzero&\mzero&\mzero&\mzero&\mzero&\mzero&\mzone&\mzone&\mzero&\mzero&\mzero&\mzero&\mzero&\mzero&\mzone&\mzero&\mzero&\mzero&\mzero&\mzero&\mzero&\mzero\\
\mzero&\mzero&\mzero&\mzero&\mzero&\mzero&\mzero&\mzero&\mzone&\mzone&\mzero&\mzero&\mzero&\mzero&\mzero&\mzero&\mzone&\mzero&\mzero&\mzero&\mzero&\mzero&\mzero&\mzero&\mzero&\mzero&\mzero&\mzero&\mzero&\mzone&\mzone&\mzero&\mzero&\mzero&\mzero&\mzero&\mzero&\mzone&\mzero&\mzero&\mzero&\mzero\\
\mzero&\mzero&\mzero&\mzero&\mzero&\mzero&\mzero&\mzero&\mzero&\mzero&\mzero&\mzone&\mzone&\mzero&\mzero&\mzero&\mzero&\mzero&\mzero&\mzone&\mzero&\mzero&\mzero&\mzero&\mzero&\mzero&\mzero&\mzero&\mzero&\mzero&\mzero&\mzero&\mzone&\mzone&\mzero&\mzero&\mzero&\mzero&\mzero&\mzero&\mzone&\mzero\\
\mzero&\mzone&\mzero&\mzero&\mzero&\mzero&\mzero&\mzero&\mzero&\mzero&\mzero&\mzero&\mzero&\mzero&\mzone&\mzone&\mzero&\mzero&\mzero&\mzero&\mzero&\mzero&\mzone&\mzero&\mzero&\mzero&\mzero&\mzero&\mzero&\mzero&\mzero&\mzero&\mzero&\mzero&\mzero&\mzone&\mzone&\mzero&\mzero&\mzero&\mzero&\mzero\\
\mzero&\mzero&\mzero&\mzero&\mzone&\mzero&\mzero&\mzero&\mzero&\mzero&\mzero&\mzero&\mzero&\mzero&\mzero&\mzero&\mzero&\mzone&\mzone&\mzero&\mzero&\mzero&\mzero&\mzero&\mzero&\mzone&\mzero&\mzero&\mzero&\mzero&\mzero&\mzero&\mzero&\mzero&\mzero&\mzero&\mzero&\mzero&\mzone&\mzone&\mzero&\mzero\\
\mhit&\mzero&\mzero&\mzero&\mzero&\mzero&\mzero&\mzone&\mzero&\mzero&\mzero&\mzero&\mzero&\mzero&\mzero&\mzero&\mzero&\mzero&\mzero&\mzero&\mzone&\mzone&\mzero&\mzero&\mzero&\mzero&\mzero&\mzero&\mzone&\mzero&\mzero&\mzero&\mzero&\mzero&\mzero&\mzero&\mzero&\mzero&\mzero&\mzero&\mzero&\mzone\\
\mzero&\mzero&\mzone&\mzone&\mzero&\mzero&\mzero&\mzero&\mzero&\mzero&\mzone&\mzero&\mzero&\mzero&\mzero&\mzero&\mzero&\mzero&\mzero&\mzero&\mzero&\mzero&\mzero&\mzone&\mzone&\mzero&\mzero&\mzero&\mzero&\mzero&\mzero&\mzone&\mzero&\mzero&\mzero&\mzero&\mzero&\mzero&\mzero&\mzero&\mzero&\mzero
\end{array}
\right].
\]
\endgroup
\end{center}

The colored backgrounds of the one entries indicate row groups, the yellow entries mark the column used in the preceding calculation, the thick vertical separator marks the branch boundary, the thin vertical separators mark the \(M\)-blocks, and the horizontal separator marks the row-group boundary.
The unit quantities represented by these separators are as follows: each row has weight \(2|M|\), and one branch block contains \(J|\F||M|\) one entries.
In this example, these numbers are \(6\) and \(63\), respectively.
\end{example}

\subsection{Regularity and Coset Certificates}\label{subsec:base-certificates}

This definition separates the combinatorial part of the construction from the choice of a particular field.
The field and subgroup determine the block length and row weight.
The next theorem records the first consequence of this construction rule: the desired regularity is automatic once the index sets have been fixed in this way.

\begin{theorem}[\texorpdfstring{$(J,L)$}{(J,L)} Regularity]\label{thm:two-branch-regularity}
Let $\HX$ and $\HZ$ be the binary matrices of Definition~\ref{def:two-branch-coset-base}.
Then $(\HX,\HZ)$ is $(J,L)$-regular with $L=2|M|$: every column of $\HX$ and $\HZ$ has weight $J$, and every row of $\HX$ and $\HZ$ has weight $L$.
\end{theorem}

\begin{proof}
By Definition~\ref{def:two-branch-coset-base}, each column is connected to exactly one row in each of the $J$ row groups on the $X$ side and exactly one row in each of the $J$ row groups on the $Z$ side.
Hence every column of $\HX$ and $\HZ$ has weight $J$.
For a fixed $X$ row $(i,r)$, each pair $(\lambda,h)$ determines a unique value
\[
  t=r-a_i^{(\lambda)}h,
\]
and hence a unique incident column.
Indeed, \(a_i^{(\lambda)}h\) is a fixed element of \(\F\), so subtraction in the field gives a single element \(t\in\F\).
Equivalently, if \(t\) and \(t'\) both satisfy \(r=t+a_i^{(\lambda)}h=t'+a_i^{(\lambda)}h\), then \(t=t'\).
Thus each $X$ row has weight $L=2|M|$.
The same argument applies to $Z$ rows.
This proves the claim.
\end{proof}

The remaining structural properties depend on the coefficient arrays.
The next two theorems separate the two finite-field certificates used in this paper.
The first theorem certifies CSS orthogonality from cross-type coset equalities.
The second theorem certifies the absence of same-type 4-cycles from same-type coset disjointness.
These hypotheses are sufficient for the construction in this paper; they are not claimed to be necessary for every possible two-branch base.

\begin{theorem}[CSS Orthogonality]\label{thm:orthogonality-coset-condition}
Let $\HX$ and $\HZ$ be the binary matrices of Definition~\ref{def:two-branch-coset-base}.
Assume that
\begin{equation}
\label{eq:cross-nonzero}
  b_j^{(\lambda)}-a_i^{(\lambda)}\ne0
\end{equation}
for all $\lambda\in\{0,1\}$ and all $i,j$.
Assume also that
\begin{equation}
\label{eq:cross-coset-equality}
  (b_j^{(0)}-a_i^{(0)})M=(b_j^{(1)}-a_i^{(1)})M
  \quad\text{for all } i,j.
\end{equation}
Then $\HX\HZ^{\mathsf T}=0$ over $\F_2$.
\end{theorem}

\begin{proof}
By the incidence equations \eqref{eq:hx-incidence} and \eqref{eq:hz-incidence} in Definition~\ref{def:two-branch-coset-base}, an $X$ row $(i,r)$ and a $Z$ row $(j,s)$ share a branch \(\lambda\) column $(\lambda,t,h)$ exactly when
\[
  r=t+a_i^{(\lambda)}h,\qquad
  s=t+b_j^{(\lambda)}h.
\]
Eliminating \(t\) gives
\[
  s-r=(b_j^{(\lambda)}-a_i^{(\lambda)})h
\]
for some $h\in M$.
Conversely, if this equation holds for some \(h\in M\), then setting \(t=r-a_i^{(\lambda)}h\) gives a common column satisfying both \eqref{eq:hx-incidence} and \eqref{eq:hz-incidence}.
For this fixed \(h\), the value of \(t\) is unique by the same field calculation: any \(t'\) with \(r=t'+a_i^{(\lambda)}h\) must equal \(r-a_i^{(\lambda)}h=t\).
The nonzero-difference condition \eqref{eq:cross-nonzero} makes such an $h$ unique when it exists.
Hence the coset equality \eqref{eq:cross-coset-equality} is a sufficient condition making the branch \(0\) and branch \(1\) overlaps identical as subsets of possible differences $s-r$.
Every $X$-$Z$ row overlap is therefore either absent in both branches or present once in each branch, so the binary inner product is even.
\end{proof}

The next theorem applies the same coset viewpoint to row pairs of the same type.
Its role is to guarantee that two \(X\)-rows, or two \(Z\)-rows, never share two columns.

\begin{theorem}[Sufficient Coset Criterion for Same-Type 4-Cycle Exclusion]\label{thm:same-type-4cycle-coset-condition}
Let $\HX$ and $\HZ$ be the binary matrices of Definition~\ref{def:two-branch-coset-base}.
Assume that
\begin{equation}
\label{eq:same-type-nonzero}
  a_{i'}^{(\lambda)}-a_i^{(\lambda)}\ne0,\qquad
  b_{j'}^{(\lambda)}-b_j^{(\lambda)}\ne0
\end{equation}
for all $\lambda\in\{0,1\}$, all $i<i'$, and all $j<j'$.
If, in addition,
\begin{equation}
\label{eq:x-same-type-disjoint}
  (a_{i'}^{(0)}-a_i^{(0)})M
  \cap
  (a_{i'}^{(1)}-a_i^{(1)})M
  =\varnothing
  \quad\text{for all } i<i',
\end{equation}
and
\begin{equation}
\label{eq:z-same-type-disjoint}
  (b_{j'}^{(0)}-b_j^{(0)})M
  \cap
  (b_{j'}^{(1)}-b_j^{(1)})M
  =\varnothing
  \quad\text{for all } j<j'.
\end{equation}
Then the Tanner graphs of $\HX$ and $\HZ$ have no same-type 4-cycles.
\end{theorem}

\begin{proof}
By the \(X\)-side incidence equation \eqref{eq:hx-incidence}, two distinct $X$ rows $(i,r)$ and $(i',r')$ share a branch \(\lambda\) column $(\lambda,t,h)$ exactly when
\[
  r=t+a_i^{(\lambda)}h,\qquad
  r'=t+a_{i'}^{(\lambda)}h.
\]
Eliminating \(t\), this is equivalent to
\[
  r'-r=(a_{i'}^{(\lambda)}-a_i^{(\lambda)})h.
\]
If $i=i'$, then distinctness of the two rows gives $r\ne r'$, so the displayed equation has no solution.
Thus same-group row pairs share no column.
It remains to consider $i\ne i'$; after relabeling, assume $i<i'$.
The nonzero-difference condition \eqref{eq:same-type-nonzero} then makes $h$ unique when it exists, so a fixed branch contributes at most one common column.
Across the two branches, the disjointness condition \eqref{eq:x-same-type-disjoint} prevents this row pair from sharing one column in each branch.
Thus no $X$-side row pair shares two columns.
The $Z$-side argument is identical using the disjointness condition \eqref{eq:z-same-type-disjoint}.
\end{proof}

\subsection{Finite Search and Coefficient Examples}\label{subsec:coefficient-search}

The preceding two theorems reduce the coefficient choice to a finite quotient-coset feasibility problem.
The next proposition records this feasibility condition explicitly.
It is a necessary and sufficient condition for the coefficients to satisfy the two sufficient certificates used in this construction.

\begin{proposition}[Coefficient Feasibility for the Coset Certificates]\label{prop:coefficient-feasibility}
Let \(q=|\F|\), \(m=|M|\), and let
\[
  \pi:\F^\times\to \F^\times/M
\]
be the quotient map.
For the two-branch construction to have row weight \(L\) through this certificate, one must have \(L=2m\).
Thus \(L\) is even and \(m\mid(q-1)\).
The nonzero-difference assumptions also imply that, in each branch, the \(2J\) elements
\[
  a_0^{(\lambda)},\ldots,a_{J-1}^{(\lambda)},
  b_0^{(\lambda)},\ldots,b_{J-1}^{(\lambda)}
\]
are distinct, so \(q\ge 2J\) is necessary.
If \(J\ge2\), the same-type disjointness condition further requires at least two nonzero \(M\)-cosets, equivalently \((q-1)/m\ge2\).

For a fixed pair \((\F,M)\), a choice of coefficients satisfies the hypotheses of Theorems~\ref{thm:orthogonality-coset-condition} and~\ref{thm:same-type-4cycle-coset-condition} if and only if all differences appearing below are nonzero and
\[
  \pi(b_j^{(0)}-a_i^{(0)})
  =
  \pi(b_j^{(1)}-a_i^{(1)})
  \quad\text{for all }i,j,
\]
while
\[
  \pi(a_{i'}^{(0)}-a_i^{(0)})
  \ne
  \pi(a_{i'}^{(1)}-a_i^{(1)})
  \quad\text{for all }i<i',
\]
and
\[
  \pi(b_{j'}^{(0)}-b_j^{(0)})
  \ne
  \pi(b_{j'}^{(1)}-b_j^{(1)})
  \quad\text{for all }j<j'.
\]
\end{proposition}

\begin{proof}
The identity \(L=2m\) follows from Theorem~\ref{thm:two-branch-regularity}.
A subgroup \(M\le\F^\times\) of order \(m\) can exist only when \(m\mid(q-1)\).
The nonzero-difference conditions \eqref{eq:cross-nonzero} and \eqref{eq:same-type-nonzero} say, in particular, that no two elements among the \(a_i^{(\lambda)}\)'s and \(b_j^{(\lambda)}\)'s in the same branch coincide.
This gives \(q\ge2J\).
When \(J\ge2\), the same-type condition requires two nonzero differences to lie in disjoint nonzero cosets, so at least two nonzero \(M\)-cosets are necessary.

For the exact quotient condition, recall that two nonzero elements \(x,y\in\F^\times\) generate the same \(M\)-coset precisely when \(\pi(x)=\pi(y)\).
Therefore the cross-type coset equality \eqref{eq:cross-coset-equality} is exactly the displayed equality for \(\pi(b_j^{(\lambda)}-a_i^{(\lambda)})\).
Similarly, the same-type coset disjointness conditions \eqref{eq:x-same-type-disjoint} and \eqref{eq:z-same-type-disjoint} are exactly the displayed inequalities for the corresponding \(a\)- and \(b\)-differences.
Together with the stated nonzero-difference requirement, this is equivalent to satisfying the two coset certificates.
\end{proof}

Proposition~\ref{prop:coefficient-feasibility} is a finite feasibility test, not an automatic existence theorem for every \((J,L)\).
Nevertheless, the same quotient-coset certificate is not specific to the \((3,10)\) instance.
We used the following normalized exhaustive search.
Adding a field element to all \(a_i^{(\lambda)}\) and \(b_j^{(\lambda)}\) within one branch preserves all differences in that branch, and multiplying all coefficients in both branches by one nonzero field element preserves all quotient-coset equalities and disjointness tests.
Thus, without loss of generality for this search, we set
\[
  a_0^{(0)}=a_0^{(1)}=0,\qquad a_1^{(0)}=1.
\]
After \(\bm a^{(0)}\) and \(\bm b^{(0)}\) in branch \(0\), and \(\bm a^{(1)}\) in branch \(1\), are fixed, the orthogonality coset equations restrict each \(b_j^{(1)}\) to the finite intersection
\[
  \bigcap_{i=0}^{J-1}
  \left(a_i^{(1)}+(b_j^{(0)}-a_i^{(0)})M\right).
\]
We then test the remaining same-type disjointness conditions.
For fixed \((\F,M)\), this gives a deterministic search of order
\[
  O\!\left((q-2)_{2J-2}(q-1)_{J-1}J^2|M|\right),
\]
where \((u)_v=u(u-1)\cdots(u-v+1)\).
If a coefficient choice satisfying the quotient-coset certificates exists, this normalized exhaustive search finds one up to the stated translation and scaling symmetries.
This completeness statement is only for the sufficient certificates in Proposition~\ref{prop:coefficient-feasibility}; it does not rule out other CSS bases certified by different arguments.

The use of a finite field is a convenient sufficient framework rather than a logically minimal requirement of the incidence argument.
The regularity proof uses only additive translations, while the orthogonality and 4-cycle proofs use the fact that multiplication by each nonzero coefficient difference acts injectively on the subgroup coordinate \(h\in M\).
Equivalently, the relevant difference orbits must have trivial stabilizer under the action replacing multiplication by \(M\).
Thus the same proof strategy should extend to finite additive groups equipped with a semiregular automorphism group, or to finite rings and modules after restricting to coefficient differences for which the chosen unit action is injective.
In such a generalization, the quotient-coset tests in Proposition~\ref{prop:coefficient-feasibility} would be replaced by orbit-equality and orbit-disjointness tests.
We leave this group-action version of the construction for future work; the present paper stays with finite fields because they make the certificates and the normalized exhaustive search explicit and reproducible.

Table~\ref{tab:coefficient-feasibility-examples} lists coefficient examples found by this search over small finite fields.
For the \(\F_9\) row we write \(\F_9=\F_3[\alpha]/(\alpha^2+1)\); in the prime-field rows, the entries are integers modulo \(q\).
Since the multiplicative group \(\F^\times\) of a finite field is cyclic, the subgroup \(M\le\F^\times\) is uniquely determined by the field \(\F\) and its order \(|M|=L/2\).
For this reason, the table does not list the elements of \(M\).
The parameter column gives the actual base CSS parameters computed from binary ranks:
\(n=2q|M|\) and \(k=n-\operatorname{rank}_{\F_2}(\HX)-\operatorname{rank}_{\F_2}(\HZ)\).
It also gives the verified information on the base CSS distance \(d=\min(d_X,d_Z)\).
If the lower and upper bounds agree, the entry is written as \([[n,k,d]]\).
If a useful small upper-bound witness is available, the entry is written as \([[n,k,d_{\mathrm L}\le d\le d_{\mathrm U}]]\); if only the certified lower bound is reported, it is written as \([[n,k,d\ge d_{\mathrm L}]]\).
Lower bounds come from exhaustive exclusion of lower-weight nontrivial logical representatives.
Displayed upper bounds come from explicit logical representatives; large upper-bound witnesses are omitted from the table because they give little finite-length information.
The \(N_6\) column gives the directly enumerated numbers of same-type simple base 6-cycles \((N_6(X),N_6(Z))\).
This number is also relevant for the later lift design: each same-type base 6-cycle contributes one nonzero congruence constraint on the CPM exponents, so a larger \(N_6\) can make the lift-label search harder even when the base-level certificates hold.
The \(N_{XZ}^{(2)}\) column gives the number of \(X\)-row/\(Z\)-row pairs sharing two base columns.
For every row, direct enumeration also checks that each \(X\)-row/\(Z\)-row pair shares either zero or two base columns, as required later for the CPM-lift orthogonality condition.
The table certifies only the base-level properties covered by Theorems~\ref{thm:orthogonality-coset-condition} and~\ref{thm:same-type-4cycle-coset-condition}; it does not assert that the later lift constraints or decoding properties have been verified for those parameters.
The gray special row is the \((3,10)\) base over \(\F_{16}\) used for the 64-fold CPM lift; it is listed separately from the small-field row for the same \((J,L)\).

\begingroup
\tiny
\setlength{\tabcolsep}{3pt}
\begin{longtable}{c c c c c c >{\raggedright\arraybackslash}p{0.32\linewidth}}
\caption{Coefficient examples satisfying Proposition~\ref{prop:coefficient-feasibility} over small finite fields for selected \((J,L)\). The gray special row is the \(\F_{16}\) base used for the 64-fold lift.}
\label{tab:coefficient-feasibility-examples}\\
\toprule
\((J,L)\) & field & \(|M|\) & actual parameters & \(N_{XZ}^{(2)}\) & \(N_6\) & coefficient arrays \\
\midrule
\endfirsthead
\caption[]{Coefficient examples satisfying Proposition~\ref{prop:coefficient-feasibility} continued.}\\
\toprule
\((J,L)\) & field & \(|M|\) & actual parameters & \(N_{XZ}^{(2)}\) & \(N_6\) & coefficient arrays \\
\midrule
\endhead
\bottomrule
\endfoot
\((3,6)\)  & \(\F_7\)  & 3 & \([[42,10,3]]\) & 189 & \((168,168)\) &
\(\bm a^{(0)}=(0,1,3),\ \bm b^{(0)}=(2,4,5),\)
\(\bm a^{(1)}=(0,3,1),\ \bm b^{(1)}=(4,2,5)\) \\
\((3,8)\)  & \(\F_9\)  & 4 & \([[72,22,6]]\) & 324 & \((432,432)\) &
\(\bm a^{(0)}=(0,1,1+\alpha),\ \bm b^{(0)}=(2,1+2\alpha,2+\alpha),\)
\(\bm a^{(1)}=(0,1+\alpha,2),\ \bm b^{(1)}=(\alpha,2+\alpha,2+2\alpha)\) \\
\((3,10)\) & \(\F_{11}\) & 5 & \([[110,48,6]]\) & 495 & \((880,880)\) &
\(\bm a^{(0)}=(0,1,2),\ \bm b^{(0)}=(3,4,5),\)
\(\bm a^{(1)}=(0,2,1),\ \bm b^{(1)}=(4,3,5)\) \\
\specialrule{0.9pt}{2pt}{0pt}
\rowcolor{gray!12}
\((3,10)\) & \(\F_{16}\) & 5 & \([[160,76,4]]\) & 720 & \((800,800)\) &
\(\bm a^{(0)}=(0,1,2),\ \bm b^{(0)}=(7,3,6),\)
\(\bm a^{(1)}=(8,13,2),\ \bm b^{(1)}=(11,10,6)\) \\
\specialrule{0.9pt}{0pt}{2pt}
\((3,12)\) & \(\F_{13}\) & 6 & \([[156,82,3]]\) & 702 & \((1560,1560)\) &
\(\bm a^{(0)}=(11,6,5),\ \bm b^{(0)}=(12,1,9),\)
\(\bm a^{(1)}=(1,4,10),\ \bm b^{(1)}=(2,11,7)\) \\
\((3,14)\) & \(\F_{29}\) & 7 & \([[406,236,6]]\) & 1827 & \((2233,2436)\) &
\(\bm a^{(0)}=(10,25,22),\ \bm b^{(0)}=(24,4,14),\)
\(\bm a^{(1)}=(1,25,26),\ \bm b^{(1)}=(18,8,14)\) \\
\((3,16)\) & \(\F_{17}\) & 8 & \([[272,174,6]]\) & 1224 & \((3808,3808)\) &
\(\bm a^{(0)}=(0,1,2),\ \bm b^{(0)}=(3,4,5),\)
\(\bm a^{(1)}=(0,3,6),\ \bm b^{(1)}=(5,8,11)\) \\
\((3,18)\) & \(\F_{19}\) & 9 & \([[342,232,6]]\) & 1539 & \((5472,5472)\) &
\(\bm a^{(0)}=(0,1,2),\ \bm b^{(0)}=(3,4,5),\)
\(\bm a^{(1)}=(0,2,5),\ \bm b^{(1)}=(10,4,7)\) \\
\((3,20)\) & \(\F_{31}\) & 10 & \([[620,438,6]]\) & 2790 & \((7130,7750)\) &
\(\bm a^{(0)}=(0,1,2),\ \bm b^{(0)}=(3,4,5),\)
\(\bm a^{(1)}=(0,5,14),\ \bm b^{(1)}=(6,30,20)\) \\
\((3,30)\) & \(\F_{31}\) & 15 & \([[930,748,6]]\) & 4185 & \((26040,26040)\) &
\(\bm a^{(0)}=(0,1,2),\ \bm b^{(0)}=(3,4,5),\)
\(\bm a^{(1)}=(0,3,6),\ \bm b^{(1)}=(11,14,4)\) \\
\((4,8)\) & \(\F_{13}\) & 4 & \([[104,6,12]]\) & 832 & \((1456,1456)\) &
\(\bm a^{(0)}=(0,1,6,5),\ \bm b^{(0)}=(12,9,10,7),\)
\(\bm a^{(1)}=(0,10,12,2),\ \bm b^{(1)}=(8,9,3,4)\) \\
\((4,10)\) & \(\F_{11}\) & 5 & \([[110,28,10]]\) & 880 & \((3520,3520)\) &
\(\bm a^{(0)}=(9,8,1,4),\ \bm b^{(0)}=(7,3,0,6),\)
\(\bm a^{(1)}=(5,10,7,3),\ \bm b^{(1)}=(6,9,2,0)\) \\
\((4,12)\) & \(\F_{13}\) & 6 & \([[156,58,8]]\) & 1248 & \((6240,6240)\) &
\(\bm a^{(0)}=(6,4,2,8),\ \bm b^{(0)}=(10,11,3,0),\)
\(\bm a^{(1)}=(9,8,4,0),\ \bm b^{(1)}=(6,1,5,11)\) \\
\((4,14)\) & \(\F_{29}\) & 7 & \([[406,180,d\ge9]]\) & 3248 & \((8526,9744)\) &
\(\bm a^{(0)}=(0,1,17,13),\ \bm b^{(0)}=(25,5,20,10),\)
\(\bm a^{(1)}=(0,28,13,15),\ \bm b^{(1)}=(23,5,16,12)\) \\
\((4,16)\) & \(\F_{17}\) & 8 & \([[272,142,8]]\) & 2176 & \((15232,15232)\) &
\(\bm a^{(0)}=(0,1,14,10),\ \bm b^{(0)}=(11,13,5,9),\)
\(\bm a^{(1)}=(0,6,16,13),\ \bm b^{(1)}=(11,1,7,4)\) \\
\((4,18)\) & \(\F_{19}\) & 9 & \([[342,196,9\le d\le10]]\) & 2736 & \((21888,21888)\) &
\(\bm a^{(0)}=(0,1,18,17),\ \bm b^{(0)}=(9,12,10,11),\)
\(\bm a^{(1)}=(0,12,11,8),\ \bm b^{(1)}=(6,10,18,7)\) \\
\((4,20)\) & \(\F_{31}\) & 10 & \([[620,378,9\le d\le10]]\) & 4960 & \((29140,29450)\) &
\(\bm a^{(0)}=(0,1,29,27),\ \bm b^{(0)}=(10,8,25,7),\)
\(\bm a^{(1)}=(0,26,28,14),\ \bm b^{(1)}=(21,2,12,19)\) \\
\((4,22)\) & \(\F_{23}\) & 11 & \([[506,328,9\le d\le10]]\) & 4048 & \((40480,40480)\) &
\(\bm a^{(0)}=(0,1,9,22),\ \bm b^{(0)}=(16,7,12,21),\)
\(\bm a^{(1)}=(0,22,20,12),\ \bm b^{(1)}=(4,7,13,10)\) \\
\((4,24)\) & \(\F_{37}\) & 12 & \([[888,598,8]]\) & 7104 & \((51060,51948)\) &
\(\bm a^{(0)}=(0,1,26,17),\ \bm b^{(0)}=(3,12,35,14),\)
\(\bm a^{(1)}=(0,7,13,18),\ \bm b^{(1)}=(24,21,25,11)\) \\
\((4,26)\) & \(\F_{53}\) & 13 & \([[1378,960,d\ge7]]\) & 11024 & \((60632,68900)\) &
\(\bm a^{(0)}=(0,1,49,2),\ \bm b^{(0)}=(4,20,12,11),\)
\(\bm a^{(1)}=(0,11,21,51),\ \bm b^{(1)}=(43,31,22,4)\) \\
\((4,28)\) & \(\F_{29}\) & 14 & \([[812,586,8]]\) & 6496 & \((84448,84448)\) &
\(\bm a^{(0)}=(0,1,8,24),\ \bm b^{(0)}=(15,11,28,17),\)
\(\bm a^{(1)}=(0,26,23,12),\ \bm b^{(1)}=(18,8,28,19)\) \\
\((4,30)\) & \(\F_{31}\) & 15 & \([[930,688,d\ge7]]\) & 7440 & \((104160,104160)\) &
\(\bm a^{(0)}=(0,1,11,6),\ \bm b^{(0)}=(23,19,17,14),\)
\(\bm a^{(1)}=(0,17,8,10),\ \bm b^{(1)}=(23,2,21,20)\) \\
\((5,10)\) & \(\F_{11}\) & 5 & \([[110,8,12]]\) & 1375 & \((8800,8800)\) &
\(\bm a^{(0)}=(4,3,9,6,2),\ \bm b^{(0)}=(1,0,10,8,5),\)
\(\bm a^{(1)}=(8,6,10,2,7),\ \bm b^{(1)}=(4,5,3,9,1)\) \\
\end{longtable}
\endgroup

The \((3,30)\) row gives a \(J=3\) base whose design-rate lower bound is \(1-2J/L=1-6/30=0.8\).
It is the smallest such row within this two-branch certificate: the condition \(1-2J/L\ge0.8\) with \(J=3\) forces \(L\ge30\), hence \(|M|=L/2\ge15\), and the same-type disjointness certificate requires at least two nonzero \(M\)-cosets, so \(q-1\ge2|M|\ge30\).
Thus \(q\ge31\), and the displayed \(\F_{31}\) example attains the minimum possible field size and base length \(n=2q|M|=930\) under these certificate conditions.

\subsection{Zero Congruence Constraints for Orthogonality}\label{subsec:cpm-lift-orthogonality}

This subsection states the condition that preserves CSS orthogonality when a base pair is replaced by a \(P\)-fold circulant permutation matrix (CPM) lift.
The condition is a zero congruence constraint on the CPM exponents.
The base-side hypothesis in Theorem~\ref{thm:cpm-lift-css-orthogonality} is not merely that each binary inner product is even.
It requires every \(X\)-row and \(Z\)-row to share either zero or two base columns.
If two columns are shared, the two CPM blocks can be made identical and hence cancel over \(\F_2\).
If four or more columns are shared, the binary base inner product is still zero, but a different condition is needed to cancel a sum of more than two CPM blocks.
The \(N_{XZ}^{(2)}\) column in Table~\ref{tab:coefficient-feasibility-examples} records the direct check of this hypothesis for each listed base.

\begin{theorem}[CSS Orthogonality after CPM Lifting]\label{thm:cpm-lift-css-orthogonality}
Let \((\HX,\HZ)\) be a binary base-matrix pair.
Assume that every \(X\)-row and \(Z\)-row share either zero or two base columns.
In a \(P\)-fold CPM lift, assign
\[
  s_X(r,c),\ s_Z(z,c)\in\mathbb Z/P\mathbb Z
\]
to every nonzero \(X\)-edge \((r,c)\) and \(Z\)-edge \((z,c)\), and replace the corresponding nonzero entries by \(\Pi^{s_X(r,c)}\) and \(\Pi^{s_Z(z,c)}\).
Here \(\Pi^s\) is the \(P\times P\) CPM whose nonzero entry in row \(u\) is in column \(u+s\).
If an \(X\)-row \(r\) and a \(Z\)-row \(z\) share exactly two base columns \(c_0,c_1\), assume that
\begin{equation}
\label{eq:cpm-lift-css-zero-constraint}
  s_X(r,c_0)-s_Z(z,c_0)
  \equiv
  s_X(r,c_1)-s_Z(z,c_1)
  \pmod P
\end{equation}
for every such pair.
Then the lifted matrices satisfy
\[
  \HXL(\HZL)^{\mathsf T}=0 .
\]
\end{theorem}

\begin{proof}
If \(X\)-row \(r\) and \(Z\)-row \(z\) share no base column, the corresponding lifted block inner product is zero.
If they share base columns \(c_0,c_1\), the lifted block inner product over \(\F_2\) is
\[
  \Pi^{s_X(r,c_0)-s_Z(z,c_0)}
  +
  \Pi^{s_X(r,c_1)-s_Z(z,c_1)} .
\]
Condition~\eqref{eq:cpm-lift-css-zero-constraint} makes the two CPMs identical, so their sum over \(\F_2\) is zero.
Thus every lifted \(X\)-row/\(Z\)-row block inner product is zero, and \(\HXL(\HZL)^{\mathsf T}=0\) follows.
\end{proof}

Equation~\eqref{eq:cpm-lift-css-zero-constraint} is a homogeneous linear congruence in the CPM exponents.
If all base-edge labels are collected into a vector \(\bm s\), the lift-orthogonality constraints have the form
\[
  A_0\bm s=\bm 0\pmod P .
\]

\subsection{Nonzero Congruence Constraints for 6-Cycle Exclusion}\label{subsec:cpm-lift-six-cycle-exclusion}

We next state the condition that prevents a same-type base 6-cycle from closing after the lift while preserving the zero congruence constraints for orthogonality.
For each base 6-cycle \(\gamma\), 6-cycle exclusion requires a linear form \(\ell_\gamma(\bm s)\) in the CPM exponents to satisfy
\[
  \ell_\gamma(\bm s)\ne0\pmod P .
\]
The method for finding labels that satisfy both zero and nonzero congruence constraints is described in Section~\ref{subsec:congruence-label-search}.

The next theorem is independent of the finite-field realization of the base.
Once the base has no same-type 4-cycles, a same-type lifted 6-cycle can appear from a base 6-cycle only when the corresponding signed CPM-exponent sum is zero.
Thus removing all lifted 6-cycles amounts to imposing one nonzero congruence constraint for each same-type base 6-cycle.
When \(N_6\) is large, these nonzero constraints may substantially restrict the remaining choices of lift labels, and satisfiability is a separate finite check rather than a consequence of the base construction alone.

\begin{theorem}[CPM Lift Condition for Avoiding Lifted 6-Cycles]\label{thm:cyclic-lift-six-cycle-criterion}
Fix one side \(S\in\{X,Z\}\) and let \(H_S\) be a binary base matrix whose Tanner graph has no 4-cycles.
For \(s\in\mathbb Z/P\mathbb Z\), let \(\Pi^s\) be the \(P\times P\) CPM whose nonzero entry in row \(u\) is in column \(u+s\), with indices read modulo \(P\).
Construct a \(P\)-fold CPM lift \(H_S^{\mathrm{lift}}\) of \(H_S\) by replacing each nonzero base entry, equivalently each base edge \((r,c)\), by \(\Pi^{s_S(r,c)}\), where
\[
  s_S(r,c)\in\mathbb Z/P\mathbb Z .
\]
Thus lifted check \((r,u)\) is incident to lifted variable \((c,u+s_S(r,c))\).
For a base 6-cycle
\[
  r_0-c_0-r_1-c_1-r_2-c_2-r_0,
\]
define its signed CPM-exponent sum by
\[
\begin{aligned}
  \Delta_S
  &=
  \bigl(s_S(r_0,c_0)-s_S(r_1,c_0)\bigr)
  +\bigl(s_S(r_1,c_1)-s_S(r_2,c_1)\bigr) \\
  &\qquad
  +\bigl(s_S(r_2,c_2)-s_S(r_0,c_2)\bigr)
  \pmod P .
\end{aligned}
\]
This base 6-cycle closes to a length-6 cycle in the lifted Tanner graph if and only if \(\Delta_S=0\).
Consequently, if \(\Delta_S\ne0\) for every same-type base 6-cycle on side \(S\), then the lifted side-\(S\) Tanner graph has girth at least \(8\).
\end{theorem}

\begin{proof}
Follow a lifted walk above the displayed base cycle, starting at lifted check \((r_0,u)\).
In the CPM block at \((r_0,c_0)\), the nonzero entry in row coordinate \(u\) is in column coordinate \(u+s_S(r_0,c_0)\).
Thus the walk reaches variable \((c_0,u+s_S(r_0,c_0))\).
To continue through the CPM block at \((r_1,c_0)\), the next lifted check coordinate must be
\[
  u+s_S(r_0,c_0)-s_S(r_1,c_0).
\]
Continuing the same calculation along the six CPM blocks changes the lifted check coordinate by exactly \(\Delta_S\).
Therefore the lifted walk returns to the starting lifted check coordinate if and only if \(\Delta_S=0\).
This is the standard cycle condition for cyclic graph lifts \cite{kelley2008ldpc}.

Because \(H_S\) has no base 4-cycles, replacing nonzero entries by CPM blocks cannot create lifted 4-cycles unless a base 4-cycle exists.
If all base 6-cycles have nonzero signed CPM-exponent sum, no 6-cycle closes in the lift.
Therefore the lifted side-\(S\) Tanner graph has no cycles of length \(4\) or \(6\), and its girth is at least \(8\).
\end{proof}

\subsection{Nonzero Congruence Constraints for Low-Weight Lift Supports}\label{subsec:lift-coordinate-coset-support-exclusion}

We first fix the terminology.
In a \(P\)-fold CPM lift, a base column \(c\) is replaced by the \(P\) lifted columns
\[
  (c,u),\qquad u\in\mathbb Z/P\mathbb Z .
\]
We call this the lifted column set associated with \(c\).

Let \(K\le\mathbb Z/P\mathbb Z\) be a subgroup.
Fix a base-column support \(T\), and choose a representative \(f_c\in\mathbb Z/P\mathbb Z\) for each \(c\in T\).
The lifted support
\[
  \widetilde T(\bm f,K)
  =
  \{(c,f_c+k):\ c\in T,\ k\in K\}
\]
is called the lift-coordinate coset lift support with base support \(T\), lift-coordinate subgroup \(K\), and representative vector \(\bm f=(f_c)_{c\in T}\).
It uses one coset \(f_c+K\) among the lift coordinates over each base column \(c\), rather than all \(P\) lifted columns associated with \(c\).
Its weight is therefore \(|T||K|\).
From the construction viewpoint, this definition separates two choices.
The base support \(T\) specifies which base columns are used, while \(\bm f\) specifies which lift coordinates are used after lifting.
If a base-stage support satisfies
\[
  \HX\vect{1}_T=\vect{0}
\]
and is not generated by \(Z\)-check rows, it can become a \(Z\)-type logical-operator candidate.
However, the base support is only the projection of a lifted support.
For a lift-coordinate coset support, the corresponding lifted equation is
\[
  \HXL\vect{1}_{\widetilde T(\bm f,K)}=\vect{0}.
\]
To become a genuine low-weight lifted support, this lifted equation must hold for some choice of representatives \(f_c\).

Thus removing lifted 6-cycles is not the only finite-length constraint needed after the base has been fixed.
One must also check specified small base-support families \(\mathcal O\) and small lift-coordinate subgroups \(K\).
In the \((3,10)\) instance below, the relevant base support has eight base columns and uses a two-point coset over each base column, giving a lifted candidate of weight \(16\).

The following theorem states only a necessary condition for this type of lift-coordinate coset support to have zero \(X\)-syndrome.
It does not characterize all logical operators.
How this necessary condition is used in lift design is explained after the theorem.

\begin{theorem}[Necessary Condition for a Lift-Coordinate Coset Support to Have Zero Syndrome]\label{thm:lift-coordinate-coset-support-exclusion}
Let \((\HX,\HZ)\) be a binary CSS base pair.
For \(s\in\mathbb Z/P\mathbb Z\), let \(\Pi^s\) be the \(P\times P\) CPM whose nonzero entry in row \(u\) is in column \(u+s\).
Let \(\HXL\) be the \(P\)-fold CPM lift of \(\HX\) obtained by replacing each nonzero \(X\)-edge \((r,c)\) by \(\Pi^{s_X(r,c)}\), where
\[
  s_X(r,c)\in\mathbb Z/P\mathbb Z .
\]
Let \(K\le \mathbb Z/P\mathbb Z\) be a subgroup.
Fix a base-column support \(T\), and consider lifted supports of the form
\[
  \widetilde T(\bm f,K)
  =
  \{(c,f_c+k):\ c\in T,\ k\in K\},
  \qquad
  f_c\in \mathbb Z/P\mathbb Z .
\]
If
\[
  \HXL\vect{1}_{\widetilde T(\bm f,K)}=\vect{0},
\]
then for each \(X\)-row \(r\) whose intersection with \(T\) contains exactly two columns \(a,b\),
\begin{equation}
\label{eq:lift-coordinate-coset-support-equation}
  f_b-f_a \equiv s_X(r,b)-s_X(r,a)
  \quad \text{in }(\mathbb Z/P\mathbb Z)/K
\end{equation}
is necessary.
\end{theorem}

\begin{proof}
In the CPM block \(\Pi^{s_X(r,c)}\), lifted check coordinate \(u\) is adjacent to lifted variable coordinate
\[
  v=u+s_X(r,c).
\]
Equivalently, the lifted check coordinate reached from \((c,v)\) through edge \((r,c)\) is \(u=v-s_X(r,c)\).
If \(r\) meets \(T\) exactly in \(a,b\), then the lifted support gives the two lifted-check-coordinate cosets
\[
  f_a+K-s_X(r,a),
  \qquad
  f_b+K-s_X(r,b).
\]
For the lifted syndrome to vanish at this row, these two cosets must be equal.
This equality is equivalent to~\eqref{eq:lift-coordinate-coset-support-equation} in the quotient group \((\mathbb Z/P\mathbb Z)/K\).
Thus the displayed equation is necessary for the prescribed pattern to have a zero-syndrome lift.
\end{proof}

The next example explains how the finite consistency test in Theorem~\ref{thm:lift-coordinate-coset-support-exclusion} is used for the selected lift.
It records a concrete support family that is excluded; it is not a proof that every weight-16 logical operator is absent.

\begin{example}[Excluding a Two-Point Lift-Coordinate Coset Pattern in the 64-Fold Lift]\label{ex:lift-coordinate-coset-weight16-exclusion}
In the \((3,10)\) construction considered here, the candidate to be excluded has eight base columns and places two lift-coordinate values over each base column.
Take
\[
  P=64,\qquad K=\{0,32\}\le \mathbb Z/64\mathbb Z,
\]
and, for example,
\[
  T_0=\{10,25,55,60,99,104,134,149\}.
\]
This is not an isolated support.
Together with the paired seed
\[
  T_1=\{15,20,50,65,94,109,139,144\},
\]
translations in the finite field and scalings by the subgroup used in this construction generate 20 distinct base supports of this form after duplicates are removed.
For each support, the relevant \(X\)-row graph is connected.
Thus, if the quotient congruences below are consistent, the representatives \(f_c\bmod 32\) are determined up to a global translation.
Consequently, one closing base support would give 32 lifted supports of weight \(16\).
The lift-coordinate coset support on \(T_0\) has the form
\[
  \{(c,f_c),(c,f_c+32):\ c\in T_0\},
\]
so its weight is \(8\cdot2=16\).
Since \((\mathbb Z/64\mathbb Z)/K\) is naturally identified with \(\mathbb Z/32\mathbb Z\), each \(X\)-row meeting \(T_0\) in exactly two columns \(a,b\) gives the necessary congruence
\[
  f_b-f_a \equiv s_X(r,b)-s_X(r,a)\pmod {32}.
\]
Thus the exclusion check is a finite consistency test on the unknowns \(f_c\bmod 32\).
If the congruences around the rows touching \(T_0\) are inconsistent, then no choice of representatives \(f_c\) can make this two-point lift-coordinate coset pattern have zero \(X\)-syndrome.
The same test is applied to all 20 supports in the orbit.
For the selected 64-fold lift, all 20 systems are inconsistent.
This excludes the specified weight-16 \(Z\)-type lift-coordinate coset pattern, but it does not by itself prove that all weight-16 logical operators are absent.
\end{example}

The theorem converts a specified low-weight lift-support pattern into finite linear constraints on lift labels.
For a fixed base support \(T\), form the constraint graph obtained from~\eqref{eq:lift-coordinate-coset-support-equation}: its vertices are the columns in \(T\), and each \(X\)-row meeting \(T\) in exactly two columns gives an edge between them.
Adding the equations around a cycle \(\eta\) in this graph cancels the \(f_c\) terms.
Thus the necessary condition includes a zero congruence
\[
  \ell_{T,\eta}(\bm s)=0
  \quad\text{in }(\mathbb Z/P\mathbb Z)/K .
\]
To exclude the specified support, it is enough to choose labels for which at least one such cycle satisfies
\[
  \ell_{T,\eta}(\bm s)\ne0
  \quad\text{in }(\mathbb Z/P\mathbb Z)/K .
\]
In this sense, exclusion of a specified low-weight lift support can also be handled as a nonzero congruence constraint on the lift labels.

\subsection{Label Search Satisfying Zero and Nonzero Congruence Constraints}\label{subsec:congruence-label-search}

The previous subsections express the lift-label requirements as congruence constraints.
Orthogonality gives zero congruence constraints
\[
  A_0\bm s=\bm 0\pmod P,
\]
whereas 6-cycle exclusion and specified low-weight lift-support exclusion give nonzero congruence constraints of the form
\[
  \ell_j(\bm s)\ne0\pmod {m_j}.
\]
Here \(m_j=P\) for 6-cycle exclusion, while for lift-coordinate coset support exclusion, \(m_j=|(\mathbb Z/P\mathbb Z)/K|\) after identifying the quotient group with \(\mathbb Z/m_j\mathbb Z\).

The search is a finite congruence-satisfaction problem.
First, all CPM exponents on base edges are collected as unknowns, and the zero congruence system \(A_0\bm s=\bm 0\) from orthogonality is solved.
Keeping this solution set, the nonzero congruence constraints are tested one by one.
When a violated constraint \(\ell_j(\bm s)=0\) is found, a nonzero value \(u_j\in\mathbb Z/m_j\mathbb Z\) is selected and the linear congruence
\[
  \ell_j(\bm s)=u_j\pmod {m_j}
\]
is added.
If the enlarged linear system becomes inconsistent, another nonzero value is tried or the previous choice is backtracked.
This iteration avoids the zero sets of the nonzero constraints while always remaining inside the orthogonality solution set.

The implementation uses random nonzero values and restarts.
Since the candidate values for each \(u_j\) are finite, the procedure can be made complete by backtracking over all candidates.
Therefore, as long as only finitely many zero and nonzero congruence constraints are considered, the complete version finds a solution whenever one exists.
For the adopted labels, all zero and nonzero congruence constraints are directly rechecked after the search, independently of the search history.

\subsection{Testing Distance Upper-Bound Witnesses}\label{subsec:distance-upper-bound-witness}

Lift constraints exclude short cycles and specified low-weight support patterns.
By contrast, an upper bound on the minimum distance of the constructed finite-length code is obtained by finding an explicit logical representative and verifying that it is not a stabilizer.
The next theorem states this test for a general CSS check-matrix pair.
Here a witness means a non-stabilizer logical representative whose kernel condition and row-space nonmembership have both been directly verified.
It is evidence for an upper bound, not a determination of the exact distance.

\begin{theorem}[Testing Distance Upper-Bound Witnesses]\label{thm:distance-upper-bound-witness}
Let \((H_X,H_Z)\) be a binary CSS check-matrix pair.
If a column set \(S_X\) satisfies
\[
  \vect{1}_{S_X}\in\kerop H_Z,
  \qquad
  \vect{1}_{S_X}\notin \row(H_X),
\]
then the corresponding CSS code satisfies
\[
  d_X\le |S_X|.
\]
Similarly, if a column set \(S_Z\) satisfies
\[
  \vect{1}_{S_Z}\in\kerop H_X,
  \qquad
  \vect{1}_{S_Z}\notin \row(H_Z),
\]
then
\[
  d_Z\le |S_Z|.
\]
If both witnesses are available, then
\[
  d\le \min\{|S_X|,|S_Z|\}.
\]
\end{theorem}

\begin{proof}
The condition \(\vect{1}_{S_X}\in\kerop H_Z\) says that \(S_X\) is a candidate \(X\)-type logical representative.
The condition \(\vect{1}_{S_X}\notin \row(H_X)\) says that this representative is not an \(X\)-stabilizer.
Thus a nontrivial \(X\)-type logical representative of weight \(|S_X|\) exists, and \(d_X\le |S_X|\).
The same argument gives \(d_Z\le |S_Z|\) for a \(Z\)-type witness.
Since \(d=\min(d_X,d_Z)\), the final claim follows.
\end{proof}

This verification rule is the same logical-witness test used in \cite{kasai2026heuristic}.
It is separate from Theorem~\ref{thm:lift-coordinate-coset-support-exclusion}.
Theorem~\ref{thm:lift-coordinate-coset-support-exclusion} excludes specified low-weight support patterns, whereas Theorem~\ref{thm:distance-upper-bound-witness} gives an upper bound by directly checking explicit logical representatives on the constructed matrices.

\section{Concrete Example of Section~\ref{sec:general-theory}: the \texorpdfstring{$\F_{16}$}{F16} Base and 64-Fold Lift Constraints}\label{sec:f16-instance}\label{sec:lift-constraints}
We now specialize the general construction and lift constraints from Section~\ref{sec:general-theory} to the instance used in this paper.
We fix the \((3,10)\) base over \(\F_{16}\) shown as the gray special row in Table~\ref{tab:coefficient-feasibility-examples}, and then state the conditions imposed on the 64-fold CPM lift.
The base-stage coefficients, base parameters, and same-type base 6-cycle counts are recorded in the table.
The following subsections correspond to Sections~\ref{subsec:cpm-lift-orthogonality}--\ref{subsec:congruence-label-search}: zero congruence constraints for orthogonality, nonzero congruence constraints for 6-cycle exclusion, nonzero congruence constraints for low-weight lift-support exclusion, and the label search satisfying all of them.

The table states the selected \(\F_{16}\) coefficients and records that the base has 800 same-type simple 6-cycles on each side.
At the lift stage, the degree, commutation, and same-type 4-cycle conditions are already fixed at the base level.
The remaining finite-length checks are therefore to keep CSS orthogonality after CPM lifting, make each same-type base 6-cycle nonclosing, and exclude the specified weight-16 logical-support orbit.

\subsection{Concrete Zero Congruence Constraints for Orthogonality}
Corresponding to Section~\ref{subsec:cpm-lift-orthogonality}, we first specialize the zero congruence constraints that preserve CSS orthogonality after the 64-fold lift.
The base construction gives CSS orthogonality by making every \(X\)-row and \(Z\)-row share either zero or two base columns.
However, orthogonality after CPM lifting also depends on the labels.
If an \(X\)-row \(r\) and a \(Z\)-row \(z\) share base columns \(c_0,c_1\), the lifted block inner product is
\[
  \Pi^{s_X(r,c_0)-s_Z(z,c_0)}
  +
  \Pi^{s_X(r,c_1)-s_Z(z,c_1)} .
\]
This sum is zero over \(\F_2\) if and only if
\[
  s_X(r,c_0)-s_Z(z,c_0)
  \equiv
  s_X(r,c_1)-s_Z(z,c_1)
  \pmod {64}.
\]
Thus the lift-label search imposes no condition on \(X\)-row/\(Z\)-row pairs with no shared base column, and imposes the displayed congruence on all pairs sharing two base columns.
For the selected 64-fold labels, this congruence holds for all 720 paired \(X\)-\(Z\) overlaps, so \(\HXL(\HZL)^{\mathsf T}=0\) is preserved.

\subsection{Concrete Nonzero Congruence Constraints for 6-Cycle Exclusion}
After the base construction, the remaining short-cycle condition concerns the same-type 6-cycles.
The CPM lift is not used to repair commutation or row weights; those properties have already been settled at the base stage.
The general acceptance condition is Theorem~\ref{thm:cyclic-lift-six-cycle-criterion}: every same-type base 6-cycle must have nonzero signed CPM-exponent sum.
The base has 800 same-type simple 6-cycles on each of the $X$ and $Z$ sides.

This requirement is satisfied by the 64-fold CPM lift used below.
The liftability verification over both $\F_2$ and $\F_{997}$ found that all 800 same-type 6-cycles are avoidable and that no 6-cycle is forced to close.
For the selected 64-fold labels, no base 6-cycle has zero signed CPM-exponent sum, so no same-type 6-cycle closes in the lifted graph.

The base length is \(160\), so the 64-fold lift has length \(N=10240\) and \(3072\) rows on each side.
It remains \((3,10)\)-regular on both sides and satisfies the CSS condition.
For the selected labels, the number of lifted same-type 6-cycles arising from base 6-cycles is zero on both the \(X\) and \(Z\) sides.

\subsection{Concrete Low-Weight Lift-Support Exclusion}
Avoiding lifted 6-cycles alone does not prevent the minimum distance from becoming small.
Even among lifts in which no same-type base 6-cycle closes, many candidates can appear that become weight-16 \(Z\)-type logical operators after lifting.
They arise when a small support of eight base columns is placed on two lift-coordinate values over each base column.
Thus the next finite-length support pattern to exclude is the following weight-16 logical orbit.

This orbit is obtained at the base stage by looking for small supports in \(\ker\HX\) that are not in \(\row(\HZ)\), and then closing the resulting candidates under the symmetries of the \(\F_{16}\) instance.
The two seed supports are
\[
  T_0=\{10,25,55,60,99,104,134,149\}
\]
and
\[
  T_1=\{15,20,50,65,94,109,139,144\}.
\]
Both lie in $\ker\HX$ but not in $\row(\HZ)$.
If they close in the lift, they produce pure $Z$-type nondegenerate logical operators.
The forbidden lifted pattern places a two-point lift-coordinate coset \(\{f,f+32\}\) over each base column.
Since each seed support has eight base columns, this pair pattern has total weight \(16\).

These two supports generate the full family considered here.
By ``orbit'' we mean the 20 distinct supports obtained from them, after duplicate supports are removed, by translations in $\F_{16}$, scalings by $M$, and the two orientations.
All 20 lie in $\ker\HX$ and outside $\row(\HZ)$.
Thus closure of this orbit in the lift would produce a weight-16 $Z$-type nondegenerate logical operator.

The exclusion uses Theorem~\ref{thm:lift-coordinate-coset-support-exclusion} with \(P=64\) and \(K=\{0,32\}\).
For each support, if the congruence system in Theorem~\ref{thm:lift-coordinate-coset-support-exclusion} is inconsistent, then that support cannot occur as a two-point lift-coordinate coset pattern in the lift.
For the selected 64-fold lift, this test excludes all 20 supports in the orbit.
This is a sufficient certificate for the specified orbit pattern; it is not a characterization of all possible low-weight logical operators.

\subsection{Concrete Label Search and Accepted Instance}

For the selected 64-fold labels, we directly rechecked all zero congruence constraints for orthogonality, all nonzero congruence constraints for same-type base 6-cycles, and all nonzero congruence constraints for the 20 weight-16 base supports.
Thus CSS orthogonality, nonclosure of same-type lifted 6-cycles, and exclusion of the specified weight-16 lift-coordinate coset pattern hold simultaneously.

The $\F_2$ ranks are
\[
  \rank(\HXL)=3066,
  \qquad
  \rank(\HZL)=3066.
\]
Thus the effective 64-fold-lift code parameters are
\[
  [[10240,4108]],
\]
with rate
\[
  R=4108/10240=0.401171875.
\]
In frame error rate (FER) plots we also use the regular $(3,10)$ design rate $R=0.4$ as a reference; this is close to, but distinct from, the finite-length effective rank.

The finite-length matrices allow a stronger direct certificate than the girth bound.
The complete support enumeration in Appendix~\ref{sec:distance-enumeration} excludes nontrivial logical representatives of weight at most $16$ on both sides.
The enumeration uses the actual FER-experiment matrices, whose coefficients are archived with this version.
This is separate from the weight-16 orbit exclusion: all supports through weight $16$ are screened, rather than one specified family.
Every kernel vector has even weight because each column has odd weight three.

\begin{proposition}[Certified distance lower bound]
For the 64-fold lift constructed here,
\[
  d_X \ge 18,
  \qquad
  d_Z \ge 18,
  \qquad
  d \ge 18.
\]
\end{proposition}

\begin{proof}
Every $X$-type logical representative is a nonzero vector in $\kerop \HZ$, and every $Z$-type logical representative is a nonzero vector in $\kerop \HX$.
Complete enumeration excludes nontrivial logical representatives of weight at most $16$ in both quotients. Odd-weight kernel vectors are impossible, so every nontrivial logical representative has weight at least $18$.
Since \(d=\min(d_X,d_Z)\), the stated CSS distance lower bound follows.
\end{proof}

The same idea gives a precise way to turn the construction into a larger certified-distance construction.
The key point is that a lower bound on the CSS distance is not obtained from decoder screening alone; it must be certified by excluding all nontrivial logical representatives below the target weight.
The following acceptance criterion is the finite-length condition that has to be added to a lift search when a larger lower bound is required.

\begin{proposition}[Target-distance acceptance criterion]\label{prop:target-distance-acceptance}
Fix a target integer \(D\).
For a lifted CSS pair \((\HX,\HZ)\), suppose that exact finite enumeration verifies
\[
  \{\,\vect{x}\in\kerop \HZ:\ 1\le \wt(\vect{x})<D\,\}
  \subseteq \row(\HX)
\]
and
\[
  \{\,\vect{z}\in\kerop \HX:\ 1\le \wt(\vect{z})<D\,\}
  \subseteq \row(\HZ).
\]
Then the corresponding CSS code satisfies
\[
  d_X\ge D,\qquad d_Z\ge D,\qquad d\ge D.
\]
\end{proposition}

\begin{proof}
The first inclusion says that every $X$-type kernel vector of weight below \(D\) is an $X$ stabilizer, so no nontrivial $X$ logical representative has weight below \(D\).
Hence \(d_X\ge D\).
The second inclusion gives the same statement for $Z$-type representatives, hence \(d_Z\ge D\).
Taking the minimum gives \(d\ge D\).
\end{proof}

This proposition is a constructive acceptance rule: in a search for lift labels, one can reject a candidate lift unless it passes the two exact inclusions for the chosen target \(D\).
The present 64-fold lift realizes Proposition~\ref{prop:target-distance-acceptance} with $D=18$, by complete support enumeration and the even-weight property.
Larger certified lower bounds require completing the same exact finite enumeration for the larger target \(D\); decoder failures or defect-focused screening by themselves are not used as distance lower-bound certificates.

We also record an explicit distance upper-bound certificate for the same finite-length code.
This step has a different purpose from the orbit exclusion above.
It does not identify the exact distance; it only records explicit logical representatives whose status has been verified exactly.
Following the acceptance criterion in \cite{kasai2026heuristic}, a candidate is used only after two binary checks: membership in the kernel of the opposite check matrix and exclusion from the corresponding stabilizer row space.
In this paragraph, a witness means such a verified non-stabilizer logical representative; it is evidence for an upper bound, not for an exact distance value.
The lift-coordinate coset exclusion condition in Theorem~\ref{thm:lift-coordinate-coset-support-exclusion} is used to rule out a specified low-weight support family.
By contrast, the following upper-bound witnesses are not consequences of that exclusion criterion; they are verified directly on the constructed lifted matrices by checking kernel membership and nonmembership in the relevant stabilizer row spaces.

\begin{proposition}[Certified distance upper bound]\label{prop:verified-distance-upper-bound}
For the 64-fold lift constructed here,
\[
  d_X \le 32,
  \qquad
  d_Z \le 32,
  \qquad
  d \le 32.
\]
\end{proposition}

\begin{proof}
For lifted columns, write \((c,f)\) for the column \(64c+f\), where \(c\) is a base column and \(f\in\mathbb Z/64\mathbb Z\).
Let \(K=16\mathbb Z/64\mathbb Z=\{0,16,32,48\}\).
For the \(X\)-type representative, take
\[
  \widetilde S_X
  =
  \{(c,r+k): (c,r)\in A_X,\ k\in K\},
\]
where
\[
  A_X=\{(2,10),(12,2),(47,13),(57,14),(108,14),(118,9),(143,14),(153,14)\}.
\]
This support has weight \(8\cdot4=32\).
Direct binary verification on the 64-fold matrices gives
\[
  \vect{1}_{\widetilde S_X}\in\kerop \HZ,
  \qquad
  \vect{1}_{\widetilde S_X}\notin \row(\HX).
\]
Thus \(d_X\le32\).

For the \(Z\)-type representative, take
\[
  \widetilde S_Z
  =
  \{(c,r+k): (c,r)\in A_Z,\ k\in K\},
\]
where
\[
  A_Z=\{(28,3),(33,0),(68,6),(73,14),(87,1),(92,7),(127,0),(132,12)\}.
\]
This support has weight \(8\cdot4=32\), and direct binary verification gives
\[
  \vect{1}_{\widetilde S_Z}\in\kerop \HX,
  \qquad
  \vect{1}_{\widetilde S_Z}\notin \row(\HZ).
\]
Thus \(d_Z\le32\).
Since the CSS distance is \(d=\min(d_X,d_Z)\), \(d\le32\).
\end{proof}

Combining the lower and upper bounds, the present certificates give
\[
  18\le d_X\le32,
  \qquad
  18\le d_Z\le32,
  \qquad
  18\le d\le32.
\]
They do not determine the exact lifted distance.

\subsection{Correction of the Published Lift Data}\label{subsec:lift-data-correction}
This version distinguishes two different $P=64$ lifts of the same base.
The FER experiments use the lift stored as \texttt{P64\_avoid\_w16\_orbit}, selected with sampler seed $20260513$ at sample $434$.
Its coefficients are supplied in \texttt{anc/fer\_distance/data/fer.csv}; its full check matrices are supplied in the same directory.
All distance bounds and the FER curve for the detailed example in this version refer to this lift.

The external coefficient file \texttt{p64\_lift\_coefficients.csv} linked by the preceding version instead encoded the earlier \texttt{P64} lift.
Its SHA-256 is
\begin{center}\footnotesize\ttfamily
137f98b600969582691a6c1912ee982e0f54\\
1c39ecdc3950cac6a0448493ebae
\end{center}
The old and selected lifts have the same length, ranks, and degree distribution, but different edge shifts.
The previous version's displayed weight-$32$ representatives were valid for the old coefficient file; they had nonzero syndromes on the FER-experiment matrices.
The representatives in Proposition~\ref{prop:verified-distance-upper-bound} have therefore been replaced by ones verified on the selected lift.

For completeness, the old coefficient file has exact CSS distance $16$.
A nontrivial $Z$-type representative is obtained from the eight base columns
\[
T_0=\{10,25,55,60,99,104,134,149\}
\]
with respective residues $(1,8,27,30,27,5,31,0)$ modulo $32$, placing both $f_c$ and $f_c+32$ over each base column.
It has zero $X$-syndrome and increases the binary rank of $H_Z$ from $3066$ to $3067$ when appended.
Complete enumeration gives $d_X\ge18$ and $d_Z\ge16$ for that old lift, so $d=d_Z=16$ there.
This support lies in the orbit excluded by the selected lift.
Consequently, the old file's distance-$16$ certificate does not apply to the FER-experiment code.
Both instances and their separate verification records are archived as ancillary material.

\section{Decoder: Joint Log-Domain BP and Low-Complexity Post-Processing}
\label{sec:decoder}
This section specifies the finite-length decoding procedure used for the constructed 64-fold lift.
It records the joint log-domain belief-propagation setting, the standard decoder parameters, and the post-processing rules used when a small residual syndrome remains.
The point of the section is to make the subsequent FER measurements reproducible as decoder measurements, separate from the construction certificates.

The construction certificates address graph structure and one known low-weight logical-support pattern.
The decoding experiment is a separate finite-length question: how this verified lifted code behaves under joint BP and deterministic post-processing.
The BP decoder used here is joint BP in the binary factor-graph formulation of CSS syndrome decoding \cite{kasai2026factorgraph}.
The choice to record small residual syndromes and post-processing outcomes follows the finite-length viewpoint that low-probability failures can be dominated by small trapping or logical structures rather than by the average ensemble behavior alone \cite{kasai2025errorfloor,komoto2025sharp,kasai2026heuristic}.
The implementation is the log-domain version: messages are stored and updated as log-likelihood ratios, which avoids probability underflow and gives numerically stable check-node and variable-node updates.
For the depolarizing channel, joint BP keeps the local correlation between the two Pauli error components through the joint prior at each qubit.
The purpose of the post-processing is to decide whether a residual syndrome has a small-neighborhood explanation visible from the joint-BP output and the lifted Tanner graph.
The decoder uses
\[
  \text{row degree}=10,\qquad
  \text{column degree}=3
\]
for the present matrices.
The standard parameters are
\[
  \text{iters}=1000,\qquad
  \text{damping}=0.3.
\]
If BP does not converge to the target syndrome, a fallback run with zero damping is attempted.

Table~\ref{tab:pp-algorithms} summarizes the post-processing rules.
They do not use the true error; they use only the visible syndrome, log-likelihood ratios, hard-decision sequence, and the lift structure.
The active rules are deliberately low-complexity post-processing rules: local linear solves on decoder-visible candidate sets, prefix-size searches over suspicious bits, flip-history candidates, short-path closures, single-column repairs, syndrome-2 core repairs, and exact or beam-style searches only for residual syndromes with one to four unsatisfied checks.
We do not run a global low-weight syndrome solver over all qubits as part of the reported decoder.
We do not count any post-hoc logical-bank correction as decoder post-processing.
Such a correction can be constructed in a replay analysis by comparing a syndrome-valid decoder output with the saved true error, extracting the residual nontrivial logical operator, and adding that same operator back to the same recorded failure.
This uses information obtained from the failure being corrected.
It is therefore used only to diagnose failure types or to certify an upper bound on the minimum distance, and it is not included in the FER measurements reported below.
For syndrome-2 residuals, we additionally look for small error-trapping-set core patterns.
Here ``core'' means the induced small residual subgraph pattern used by the repair rule.
When only one to four unsatisfied checks remain, we switch to an exact or beam-style search on decoder-visible candidate sets, mainly the corresponding small Tanner neighborhood.
Thus the implemented rule is local: it is meant to repair small residual syndromes that have a low-weight explanation visible from the BP output and the lift graph.
A small residual weight measured later by comparing with the saved true error is evidence that a stronger global low-weight syndrome solver could repair the replay case, but it is not itself information available to the decoder.

\begin{table}[H]
\centering
\small
\caption{Post-processing rules used after belief propagation. The active rules use only the decoder output, the syndrome, the log-likelihood ratios, and the lift structure. The final row records a replay-only diagnostic that is explicitly excluded from the decoder.}
\label{tab:pp-algorithms}
\begin{tabular}{L{0.18\linewidth} L{0.25\linewidth} L{0.41\linewidth}}
\toprule
Rule & Main target & Description \\
\midrule
Local linear solve &
Small residual unsatisfied syndromes &
Collects low-LLR bits near the unsatisfied checks and solves the induced linear system on that candidate set. This is the default first post-processing step. \\
Prefix-size search &
Unknown suspicious-set size &
Orders all bits by suspiciousness and uses bisection to find the smallest prefix size $K$ for which the syndrome difference is solvable. \\
Diagnostic prefix search &
Boundary between unique and non-unique local solutions &
Extends the prefix-size search by probing the transition between unique and non-unique local solution regimes. \\
Flip-history candidate set &
Oscillatory BP trajectories &
Uses the bits whose hard decisions flipped during BP iterations as the candidate set. If this history set is insufficient, it falls back to a lightweight OSD step. \\
Path-based closure &
Short lifted paths between residual checks &
Builds a candidate set from short paths connecting residual unsatisfied checks and from the closure around those paths. \\
Common-column correction &
Three unsatisfied checks explained by one column &
If three unsatisfied checks are incident to a common column, this method cancels the syndrome by flipping that single bit. \\
Syndrome-2 core repair &
Syndrome-2 error-trapping-set cores &
Matches a two-unsatisfied residual against lifted error-trapping-set templates and repairs it using the corresponding core support. \\
Small-residual exact search &
Only one to four unsatisfied checks remain &
Runs an exact or beam-style search on decoder-visible candidate sets, mainly small Tanner neighborhoods, to find a low-weight correction that cancels the small residual syndrome.  This is not a global low-weight syndrome solver, so a replay case can still remain failed even if the residual against the saved true error is small. \\
Excluded replay diagnostic &
Saved failures with known true errors &
Not used as decoder post-processing.  A nontrivial logical operator extracted by comparing a syndrome-valid replay output with the saved true error is used only as a distance upper-bound witness or failure diagnostic, never as a FER-improving correction. \\
\bottomrule
\end{tabular}
\end{table}

\section{Finite-Length Frame Error Rate Measurements}\label{sec:fer-measurements}
This section reports the measured decoding performance for the constructed code.
It defines the reference lines used in the FER figure, states how the plotted points should be interpreted, and separates these finite-sample measurements from any asymptotic threshold claim.

Figure~\ref{fig:fer-p64} summarizes the decoding measurements for the constructed 64-fold lift.
The horizontal axis is the depolarizing probability $p$, and the vertical axis is the FER.
The hashing line is computed for the effective 64-fold-lift rate
\[
  R=4108/10240=0.401171875
\]
by solving the depolarizing-channel quantum hashing equation
\[
  1-h_2(p)-p\log_2 3=R.
\]
The density-evolution (DE) line is an approximate binary-input BP density-evolution reference for the regular $(3,10)$ ensemble, not a finite-length threshold of this particular 64-fold-lift code.
The use of hashing-bound and density-evolution reference lines is consistent with recent finite-length quantum LDPC decoding studies, but the present measurements are only for the single lifted code constructed here \cite{komoto2025codingbound,kasai2025degeneracy,komoto2025sharp,kasai2026orthogonality}.
The values used here are
\[
  p_{\mathrm{hash}}(R=4108/10240)=0.09403285,
  \qquad
  p_{\mathrm{DE}}^{(3,10)}\approx0.0733.
\]

\begin{figure}[H]
  \centering
  \includegraphics[width=0.95\linewidth]{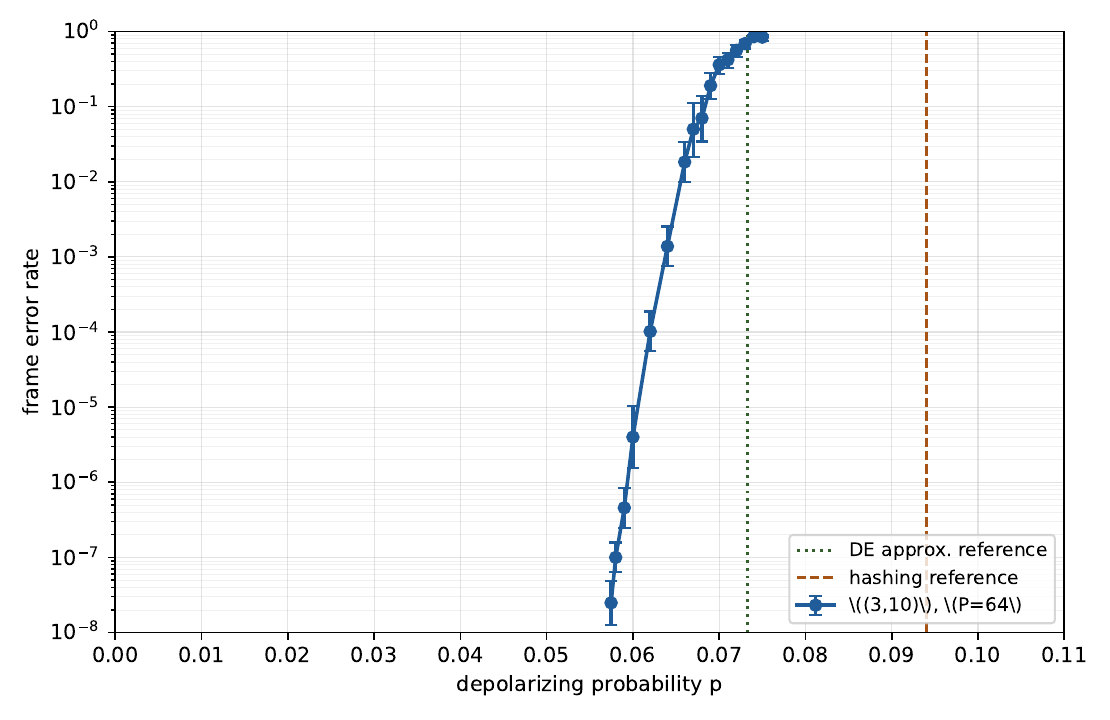}
  \caption{FER of the 64-fold lift for the \([[n,k,d]]=[[10240,4108,\,18\le d\le32]]\) CSS code. The hashing line is the depolarizing-channel quantum hashing bound for the effective rate $R=4108/10240$. The DE line is an approximate BP density-evolution reference for the regular $(3,10)$ ensemble.}
  \label{fig:fer-p64}
\end{figure}

In the $p=0.058$ run, joint BP with post-processing produced 25 recorded failures.
The run contained $180{,}000{,}000$ trials.
The full post-processing rules correct 7 of these failures by the small-residual exact-search rule, leaving 18 failures after post-processing.
The corrected FER is therefore
\[
  18/180{,}000{,}000=1.0\times10^{-7}.
\]
This value uses only the deterministic post-processing rules in Table~\ref{tab:pp-algorithms}; it does not include any post-hoc logical-bank correction extracted from the recorded failures.

At the low-$p$ side, the data include several measured points below and around \(p=0.060\), while higher-$p$ points are used to show the transition range.
These points should be read as finite-sample decoding measurements for the specific 64-fold lift and post-processing rules, not as an asymptotic threshold estimate.
The finite-sample FER measurements do not identify the minimum distance within the certified interval $18\le d\le32$.
In particular, the absence of a smaller observed logical residual is not used to infer that the distance is close to $32$.

\section{Discussion}
The paper should be read first as a base-construction paper.
The two-branch finite-field rule gives a structured way to build regular CSS LDPC bases, and the quotient-coset conditions make regularity, CSS orthogonality, and same-type 4-cycle exclusion checkable by finite computation.
Table~\ref{tab:coefficient-feasibility-examples} records that this construction method is not limited to the detailed \((3,10)\) case.

The \((3,10)\) base is the finite-length case carried through the later stages of the paper.
For that case, the base incidence pattern gives exact \((3,10)\) regularity, CSS orthogonality, and the absence of same-type 4-cycles.
The 64-fold cyclic lift prevents the remaining same-type base 6-cycles from closing, and the same lift labels are constrained to exclude the known weight-16 \(Z\)-type logical-support orbit.

These certificates do not prove the exact minimum distance.
They show that the detailed finite-length example excludes the specified short-cycle condition and the specified weight-16 support pattern, while complete enumeration through weight $16$ and certified logical witnesses give the finite bounds \(18\le d\le32\).
The decoding data then give a separate performance measurement for the verified lift under joint BP and deterministic post-processing.
The exact distance within the certified interval remains undetermined.
For the other \((J,L)\)-regular bases in Table~\ref{tab:coefficient-feasibility-examples}, the same two-branch construction already supplies the base matrices; choosing lifts, certifying distance-relevant structures, and measuring finite-length decoding performance for those parameter pairs are left for future work.
One concrete plan is to strengthen the base stage itself.
If one can construct regular CSS bases whose same-type Tanner graphs already have girth at least eight, then the lift stage no longer has to impose nonzero congruence constraints for base 6-cycle exclusion.
The cyclic-lift search would then mainly enforce the zero congruence constraints required for CSS orthogonality, together with a smaller set of nonzero constraints for explicitly targeted low-weight support patterns.
This would make the lift-label problem easier and would provide a cleaner route to larger finite-length examples.
Another direction is to remove the finite-field restriction from the base construction.
The proofs indicate that the essential object is a finite translation set with a free action on the relevant nonzero differences, not necessarily a field.
Developing the corresponding orbit-based construction over finite groups, rings, or modules may produce further regular CSS bases beyond those accessible through multiplicative subgroups of finite fields.

\appendix
\section{Computer-Assisted Distance Lower Bounds}\label{sec:distance-enumeration}
For a binary CSS pair, the distances are the minimum weights in
$\ker H_Z\setminus\operatorname{row}(H_X)$ and
$\ker H_X\setminus\operatorname{row}(H_Z)$, respectively.
The computation tests membership in the stabilizer row space by exact binary elimination.
Light stabilizers are retained in the code and are not mistaken for logical operators.
Since every column of each matrix has odd weight three, summing all check equations shows that every kernel vector has even weight.

A minimum-weight nontrivial logical support can be taken to be connected under any pairing of its incident edges at each check.
Indeed, each component of a disconnected pairing is itself a kernel support.
If all components are stabilizers, so is their sum; otherwise a smaller nontrivial component contradicts minimality.
Similarly, a minimum nontrivial support cannot strictly contain a nonzero stabilizer support, since adding that stabilizer removes its support and decreases the weight.
These observations justify restricting the search to minimal connected supports and stopping a branch when its current support is already a stabilizer.

Each column meets exactly one check in each of three check-row groups.
Pairing support variables at each check gives three colored perfect matchings on the support.
The absence of Tanner 4- and 6-cycles implies that their union is a simple triangle-free cubic graph.
This also excludes nonzero kernel supports of weights two and four.
Pairings at checks meeting four or more support variables are included.
We enumerate connected graphs of this form, retaining one representative of each color-preserving isomorphism class, and embed them into the actual check matrix.
The first two matchings are enumerated as alternating even cycles, and the third matching is enumerated over the remaining admissible pairs.
At weights $6,8,10,12,14,16$, the complete lists contain
$1,10,22,226,1838,25375$ patterns, respectively.
The embedding procedure branches over compatible columns and rejects repeated columns.
With two assigned neighbors of different colors, the absence of a Tanner 4-cycle gives at most one compatible next column.
Every completed embedding is tested for nonmembership in the stabilizer row space.

The simultaneous cyclic shift of all lift indices preserves both parity-check matrices up to row permutations.
It therefore preserves the corresponding stabilizer row spaces as well as the kernels.
The search uses one root representative per base column under this verified symmetry.
All candidate supports are checked on the full lifted matrices.
Only fully completed weight levels contribute to a lower bound; a time-limited partial search contributes none.
For the selected $P=64$ FER-experiment lift, both sides complete every pattern through weight $16$, using $160$ root representatives per side. No nontrivial logical support is found. Together with the even-weight property this proves $d_X,d_Z\ge18$.
For the earlier externally published lift, the $X$-side search completes weight $16$ and the $Z$-side search completes weight $14$ before finding a weight-$16$ logical representative. The two computations use different coefficient files and must not be conflated.
The ancillary distance records include the input coefficients, source code, complete pattern lists where applicable, completed-weight logs, and file hashes.
The lower bounds are computer-assisted exhaustive-search results.
The explicit upper-bound representatives can be verified separately by syndrome and stabilizer checks.

\begingroup
\raggedright
\bibliographystyle{IEEEtran}
\bibliography{refs}
\endgroup

\end{document}